\documentclass[12pt]{article}
\usepackage{amsfonts,amssymb,amsmath,hyperref,graphicx,color,amsthm,physics}

\makeatletter
\def\amsbb{\use@mathgroup \M@U \symAMSb}
\makeatother
\usepackage{bbold}
\newcommand{\identity}{\mathbb{1}}

\newcommand{\EtildeAB}{{\widetilde{E}_{AB}}}
\newcommand{\EtildeAC}{{\widetilde{E}_{AC}}}
\newcommand{\EtildeBC}{{\widetilde{E}_{BC}}}
\newcommand{\EtildeABC}{{\widetilde{E}_{ABC}}}

\newtheorem{theorem}{Theorem}
\newtheorem{corollary}{Corollary}

\newcommand{\be}{\begin{equation}}
\newcommand{\ee}{\end{equation}}
\newcommand{\bea}{\begin{eqnarray}}
\newcommand{\eea}{\end{eqnarray}}
\newcommand{\beas}{\begin{eqnarray*}}
\newcommand{\eeas}{\end{eqnarray*}}

\def\bra#1{\langle #1 \vert}
\def\ket#1{\vert #1 \rangle}

\begin{document}
\begin{titlepage}

\begin{center}

{\Large Entanglement groups}

\vspace{12mm}

\renewcommand\thefootnote{\mbox{$\fnsymbol{footnote}$}}
Xiaole Jiang${}^{1,2}$\footnote{xjiang2@gradcenter.cuny.edu},
Daniel Kabat${}^{1,2}$\footnote{daniel.kabat@lehman.cuny.edu},
Gilad Lifschytz${}^{3}$\footnote{giladl@research.haifa.ac.il},
Aakash Marthandan${}^{1,2}$\footnote{amarthandan@gradcenter.cuny.edu}

\vspace{6mm}

${}^1${\small \sl Department of Physics and Astronomy} \\
{\small \sl Lehman College, City University of New York} \\
{\small \sl 250 Bedford Park Blvd.\ W, Bronx NY 10468, USA}

\vspace{4mm}

${}^2${\small \sl Graduate School and University Center, City University of New York} \\
{\small \sl  365 Fifth Avenue, New York NY 10016, USA}

\vspace{4mm}

${}^3${\small \sl Department of Physics and} \\
{\small \sl Haifa Research Center for Theoretical Physics and Astrophysics} \\
{\small \sl University of Haifa, Haifa 3498838, Israel}

\end{center}

\vspace{12mm}

\noindent
We propose to define entanglement in terms of local unitary transformations acting on some parts of a
system that can be undone by local unitary transformations acting on other parts.  This leads to a characterization
of entanglement in terms of groups.  We refer to these as entanglement groups, and we refer to this notion as $g$-entanglement.  
We discuss the physical meaning of entanglement groups and contrast $g$-entanglement with other, more conventional definitions
of entanglement.
For pure states, entanglement groups are constructed as certain quotients of the stabilizer group and its subgroups. 
For mixed states, entanglement groups can be constructed from stabilizers of the purification.
We analyze the structure of entanglement groups, show that they have properties which correspond to monogamy of entanglement, and explore the restrictions placed by separability.
We show that $g$-entanglement underlies several well-known quantum tasks.

\end{titlepage}
\setcounter{footnote}{0}
\renewcommand\thefootnote{\mbox{\arabic{footnote}}}

\hrule
\tableofcontents
\bigskip
\hrule

\addtolength{\parskip}{8pt}

%%%%%%%%%%%%%%%%%%%%%%%
\section{Introduction\label{sect:intro}}
%%%%%%%%%%%%%%%%%%%%%%%
Entanglement is an essential aspect of quantum mechanics.  Entangled states violate the Bell and CHSH inequalities, and are used to perform quantum
tasks such as teleportation and dense coding.  Despite its importance, entanglement remains a somewhat elusive concept.  It is not straightforward
to characterize or classify entanglement, or to give it an operational meaning.  For reviews of approaches that have been considered in the literature, see
\cite{Horodecki:2009zz,walter2017multipartite}.

To illustrate the point, entanglement is frequently given a negative definition, by saying that a tensor product
state is not entangled.  Is there an affirmative definition of entanglement?  What is the operational significance of having a state that is not a tensor product?

Here we develop the idea that entanglement can be understood in terms of unitary transformations that act on some parts of a system,
but that are nonetheless equivalent to unitary transformations that act on other parts.  Consider a quantum system which for
definiteness could be a collection of spins.  We fix a partition of the spins into subsystems $A$, $B$, $C$, $\ldots$.
The Hilbert space is a tensor product
\be
{\cal H} = {\cal H}_A \otimes {\cal H}_B \otimes {\cal H}_C \otimes \cdots
\ee
We also fix a pure state $\ket{\psi} \in {\cal H}$.  The group of local unitary transformations consists of tensor products $u_A \otimes u_B \otimes u_C \cdots$
of unitaries that act on the individual Hilbert spaces.  Some of these transformations may be stabilizers, meaning that they leave the state invariant up to a phase.
\be
\label{stabilize}
\big(u_A \otimes u_B \otimes u_C \otimes \cdots \big) \ket{\psi} = e^{i \theta} \ket{\psi}
\ee
The collection of such unitaries defines the (local unitary) stabilizer group.  We can also regard (\ref{stabilize}) as saying that a local unitary transformation on some of the subsystems
is equivalent to a local unitary transformation on the complement.  We propose this as the defining feature of entanglement.
It leads to a precise mathematical characterization of multi-party entanglement in terms of entanglement groups, constructed
as quotients of the stabilizer group and its subgroups.  We define these groups in section \ref{sect:groups}.

This definition provides an operational sense in which entanglement enables ``spooky action at a distance:'' a unitary transformation on one part of the
system is equivalent to, or can be undone by, a unitary transformation on the complement.
This property of entanglement for bipartite systems has been appreciated since the very beginning \cite{Schrodinger_1935} and has appeared in many places in the literature
(see for example \cite{Linden_1998,Carteret_1999,Carteret_2000,Vollbrecht_2001,Zurek1,Zurek2,Walck_2007,Tzitrin:2019fwe,Bernards:2022qhd}).  For some recent applications
see \cite{Bryan:2018zsx,PhysRevA.104.022426,Noller_2023}.  It's worth noting that, although stabilizers are generic for bipartite systems, they are rare
for multi-partite systems.  We return to this in sections \ref{sect:physical} and \ref{sect:classifying}, where we point out that a single state with an
enhanced multi-party stabilizer (an enhanced entanglement group) can govern the behavior of nearby states, even though the nearby states
do not have enhanced stabilizers.

As will become clear, this group theory definition of entanglement, which we will call $g$-entanglement in situations where there is any possibility of confusion,
differs in some respects from the usual definition of entanglement, which we will refer to as $u$-entanglement when there is a need to disambiguate.
Given two different notions of entanglement, the question of which notion is relevant or useful may depend on the task at hand.  As we show in
section \ref{section:qtasks}, the notion of $g$-entanglement plays a role in the success of many quantum tasks.  For example, for the GHZ state, all two-qubit
reduced density matrices are separable, so there is no two-qubit $u$-entanglement in the GHZ state.  However we will see that the two-qubit entanglement
groups are non-trivial, meaning there is two-qubit $g$-entanglement, although the associated stabilizers take a very special and restricted form.  To see the physical relevance
of $g$-entanglement, the analysis in section \ref{subsection:qpt} makes it is clear that stabilizers of pairs of qubits (together with 3-qubit stabilizers) are responsible for winning the GHZ game.

After defining entanglement groups in section \ref{sect:groups}, we explore the structure of the entanglement groups in section \ref{sect:monogamy} and show how monogamy of entanglement appears.  In section \ref{sect:MixedStates} we turn to mixed states.  We show how to associate entanglement groups with density matrices and explore their structure.
In section \ref{sect:separable} we consider separable states and show what separability implies for $g$-entanglement.  In section \ref{section:qtasks} we analyze some well-known quantum tasks and show that
our characterization of entanglement underlies the ability to perform these tasks. We conclude in section \ref{sect:discussion} with a
discussion of related approaches and ideas for further development.  The appendices contain supporting material:
a general treatment of bipartite systems in appendix \ref{appendix:bipartite}, an extension of the Goursat lemma in appendix \ref{appendix:Goursat}, and properties and
examples of mixed-state entanglement groups in appendices \ref{appendix:twoparty}, \ref{appendix:Werner}, \ref{appendix:UBC}.

We will use the following notation throughout the paper.
\[
\begin{array}{llll}
\subset & \quad \hbox{\rm subgroup} \hspace{3cm} & \approx & \quad \hbox{\rm isomorphic groups} \\
\triangleleft & \quad \hbox{\rm normal subgroup} & Z(G) & \quad \hbox{\rm center of a group} \\
\, \cdot & \quad \hbox{\rm group product} & C(G) & \quad \hbox{\rm centralizer of a group} \\
\times & \quad \hbox{\rm direct product} & 1 & \quad \hbox{\rm identity element of a group} \\
\rtimes & \quad \hbox{\rm semidirect product} & \varnothing = \lbrace 1 \rbrace & \quad \hbox{\rm trivial group} \\
& & {\bf I} & \quad \hbox{\rm unit matrix}
\end{array}
\]

%%%%%%%%%%%%%%%%%%%%%%%
\section{Entanglement groups\label{sect:groups}}
%%%%%%%%%%%%%%%%%%%%%%%
To fix ideas, we begin by defining entanglement groups for bipartite systems, and illustrate with the simple example of a two-qubit system.
Then we define entanglement groups for tripartite systems, which brings in a few new features, and generalize to multi-partite systems.  We discuss the physical meaning of the
entanglement groups, and present an approach to classifying the entanglement of a state by considering all possible partitions into subsystems.
Finally we compare the notion of $g$-entanglement to the more usual notion of $u$-entanglement.

%%%%%%%%%%%%%%%%%%%%%%%
\subsection{Bipartite systems\label{sect:bipartite}}
%%%%%%%%%%%%%%%%%%%%%%%
For a bipartite system the Hilbert space is a tensor product, ${\cal H}={\cal H}_{A} \otimes {\cal H}_{B}$.  We denote the dimensions $d_A$, $d_B$
respectively.  We fix a pure state $\ket{\psi} \in {\cal H}$ and ask for the stabilizer group
\be
S_{AB} \subset U(d_A) \times U(d_B)
\ee
that leaves the state invariant up to a phase.
\be
\label{uAuB}
S_{AB} = \left\lbrace u_A \otimes u_B \in U(d_A) \times U(d_B) \, : \, \big( u_A \otimes u_B \big) \ket{\psi} = e^{i \theta} \ket{\psi}\right\rbrace
\ee
By definition, $S_{AB}$ consists of unitary transformation on system $A$ that can be undone by a unitary transformation on system $B$.
Intuitively we may regard such stabilizers as characterizing the entanglement between $A$ and $B$.  However we must recognize that
not all such stabilizers are due to entanglement.  This is because $S_{AB}$ has subgroups $S_A$, $S_B$ that only act on the first or
second factor of the Hilbert space.
\bea
\label{uA}
&& S_A = \left\lbrace u_A \otimes {\bf I}_B \, : \, \big( u_A \otimes {\bf I}_B \big) \ket{\psi} = e^{i \alpha} \ket{\psi}\right\rbrace \\
\label{uB}
&& S_B = \left\lbrace {\bf I}_A \otimes u_B \, : \, \big( {\bf I}_A \otimes u_B \big) \ket{\psi} = e^{i \beta} \ket{\psi}\right\rbrace
\eea
These are one-party stabilizers, that only act on one of the systems, and do not require a compensating transformation on the other system
to leave the state invariant.  As such they do not reflect entanglement between $A$ and $B$.\footnote{Indeed for a tensor product state $\vert \alpha \rangle \otimes \vert \beta \rangle$ we would
have $S_A = U(1) \times U(d_A - 1)$, $S_B = U(1) \times U(d_B - 1)$, where the $U(1)$'s reflect the freedom to multiply the state by a phase.}

We exclude these one-party stabilizers from entanglement as follows.  The groups $S_A$ and $S_B$ are invariant under conjugation by an element of
$S_{AB}$, so they are normal subgroups, $S_A \, \triangleleft \, S_{AB}$, $S_B \, \triangleleft \, S_{AB}$.  Note that $S_A$ and $S_B$ commute
element-by-element since they act on different systems.  So we have two commuting normal subgroups and can form the quotient group
\be
E_{AB} = S_{AB} / \big( S_A \times S_B \big)
\ee
This defines the two-party entanglement group which we take to characterize entanglement between $A$ and $B$.  Non-trivial elements of this group, i.e.\ elements other than the identity,
reflect the entangled nature of the state.

We give a complete treatment of entanglement groups for bipartite systems in appendix \ref{appendix:bipartite}.  Here, to give a flavor, we quote the
results for a two-qubit state.  The possibilities are:
\begin{itemize}
\item tensor product state: \, $\vert 00 \rangle$ \\[2pt]
The entanglement group is trivial, $E_{AB} = \lbrace 1 \rbrace$.
\item generic state: \, $a \vert 00 \rangle + b \vert 11 \rangle$ with $a > b > 0$ \\[2pt]
The entanglement group is $E_{AB} = U(1)$.  Elements of the entanglement group can be represented as local unitaries of the form
\be
\left(\begin{array}{cc} 1 & 0 \\ 0 & e^{i \alpha} \end{array}\right) \otimes \left(\begin{array}{cc} 1 & 0 \\ 0 & e^{-i \alpha} \end{array}\right)
\qquad \alpha \approx \alpha + 2 \pi
\ee
\item Bell state: \, ${1 \over \sqrt{2}} \big(\vert 00 \rangle + \vert 11 \rangle\big)$ \\[2pt]
The Bell state is (unitarily equivalent to) a spin singlet, so one might guess that the entanglement group is $SU(2)$, embedded in the local unitary group as
$u \otimes u^*$, $u \in SU(2)$.  However $u$ and $-u$ have the same effect on the state, so in fact the entanglement group is $PSU(2) = SU(2) / {\mathbb Z}_2$.
\end{itemize}
For a complete treatment of bipartite entanglement see appendix \ref{appendix:bipartite}.

%%%%%%%%%%%%%%%%%%%%%%%
\subsection{Tripartite systems\label{sect:tripartite}}
%%%%%%%%%%%%%%%%%%%%%%%
We now turn to tripartite systems, which bring in some new features.
Consider a Hilbert space which is a tensor product ${\cal H} = {\cal H}_A \otimes {\cal H}_B \otimes {\cal H}_C$ of factors with dimensions
$d_A$, $d_B$, $d_C$ respectively.  Fix a pure state $\ket{\psi} \in {\cal H}$.  We are interested in the stabilizer group
\be
S_{ABC} \subset U(d_A) \times U(d_B) \times U(d_C)
\ee
that preserves the state up to a phase.  That is, we are interested in local unitary transformations such that\footnote{We are defining
stabilizers using a direct product of unitary groups, so strictly speaking they are ordered triples of unitary matrices $(u_A,u_B,u_C)$ in which we keep track of which systems the various $U(1)$ phases act on.
This information gets blurred in the tensor product (\ref{unitary}) which however we will often write for convenience.
The phases of these unitary matrices are not physical and will cancel
out in the group quotients we construct.  An alternate approach would be to eliminate these phases from the beginning and define
stabilizers in terms of projective unitary groups $PSU(n) = U(n) / U(1) = SU(n) / {\mathbb Z}_n$.  As long as one sticks to a consistent
approach the entanglement groups are the same.\label{footnote:alternate}}
\be
\label{unitary}
\big(u_A \otimes u_B \otimes u_C \big) \ket{\psi} = e^{i \theta} \ket{\psi} \quad \hbox{for some phase $\theta$}
\ee
$S_{ABC}$ may have subgroups that act trivially on some factors of the Hilbert space.  For example there may be two-party stabilizers
$S_{AB},S_{AC},S_{BC}$, or one-party stabilizers $S_A,S_B,S_C$, defined for example by
\bea
\nonumber
&& S_{AB} = \left\lbrace \, (u_A,u_B,{\bf I}_C) \in S_{ABC} \, \right\rbrace \\
&& S_{A} = \left\lbrace \, (u_A,{\bf I}_B,{\bf I}_C) \in S_{ABC} \, \right\rbrace
\eea
It's straightforward to check that these are normal subgroups, for example $S_A$ is invariant under conjugation by an element
of $S_{AB}$.  So $S_A \, \triangleleft \, S_{AB}$, $S_A \, \triangleleft \, S_{ABC}$, $S_{AB} \, \triangleleft \, S_{ABC}$, etc.

We define entanglement relative to the partition $A$, $B$, $C$ as local unitary transformations that preserve the state up to a phase,
but that cannot be built up from one-party stabilizers.  According to this definition entanglement is characterized by the quotient
\be
F_{ABC} = S_{ABC} / (S_A \times S_B \times S_C)
\ee
which we will refer to as the full entanglement group.  It includes both two-party and three-party entanglement.
The study of entanglement becomes the study of the structure of the group $F_{ABC}$.

To separate out two-party entanglement, we propose to define two-party entanglement between $A$ and $B$ in terms of two-party stabilizers:
unitary transformations on $A$ that can be un-done by unitary transformations on $B$,
without acting on system $C$.  However, just as in section \ref{sect:bipartite}, we exclude stabilizers that only act on one of the systems ($A$ or $B$)
since they do not reflect entanglement between $A$ and $B$.  Thus two-party entanglement between $A$ and $B$ is characterized by the quotient group
\be
E_{AB} = S_{AB} / (S_A \times S_B)
\ee
We likewise characterize two-party entanglement between $A$ and $C$, and between $B$ and $C$, using the quotient groups
\bea
\nonumber
E_{AC} = S_{AC} / (S_A \times S_C) \\[5pt]
E_{BC} = S_{BC} / (S_B \times S_C)
\eea
We take non-trivial elements of these groups, i.e.\ elements other than the identity element, to reflect two-party entanglement between the corresponding systems.

These two-party entanglement groups are normal subgroups of the full entanglement group, for example $E_{AB} \, \triangleleft \, F_{ABC}$.  To see this, note that
a product of normal subgroups is normal, so
\be
\label{SABSC}
(S_{AB} \times S_C) \,\triangleleft\, S_{ABC}
\ee
Both sides of (\ref{SABSC}) have $S_A \times S_B \times S_C$ as a normal subgroup, so by Noether's third isomorphism theorem \cite{Robinson_1996}
\be
{S_{AB} \times S_C \over S_A \times S_B \times S_C} \, \triangleleft \, {S_{ABC} \over S_A \times S_B \times S_C}
\ee
We identify the left hand side with $E_{AB}$ which means
\be
E_{AB} \, \triangleleft \, F_{ABC}
\ee
We thus have three normal subgroups of the full entanglement group $E_{AB}$, $E_{AC}$, $E_{BC}$.
We can take their group product to obtain another normal subgroup, which we view as characterizing all possible transformations that
can be built up from two-party entanglement.\footnote{This is a normal subgroup of $F_{ABC}$ since the group product of normal subgroups is normal.  As we discuss in section
\ref{sect:noshare} distinct two-party entanglement groups can generate overlapping transformations, for example $E_{BC} \cap (E_{AB} \cdot E_{AC}) \not= \lbrace \, (1,1,1) \, \rbrace$.
This means that in general the product in (\ref{TwoParty}) is a group product and not a direct product.\label{footnote}}
\be
\label{TwoParty}
(\hbox{\rm group generated by two-party entanglement} )= E_{AB} \cdot E_{AC} \cdot E_{BC} \, \triangleleft \, F_{ABC}
\ee

\noindent
We define three-party entanglement between $A$, $B$, $C$ as non-trivial transformations that
\begin{itemize}
\item act on all three systems
\item leave the overall state invariant
\item cannot be built up from two-party entanglement
\end{itemize}
According to this definition three-party entanglement is characterized by the group
\be
\label{E1}
E_{ABC} = F_{ABC} / (E_{AB} \cdot E_{AC} \cdot E_{BC})
\ee
We take non-trivial elements of this group, i.e.\ elements other than the identity, to reflect three-party entanglement between systems $A$, $B$ and $C$.
The group can also be expressed in terms of stabilizers as
\be
\label{E2}
E_{ABC} = S_{ABC} / (S_{AB} \cdot S_{AC} \cdot S_{BC})
\ee
The equivalence of (\ref{E1}) and (\ref{E2}) follows from the third isomorphism theorem.\footnote{Divide the numerator and denominator
of (\ref{E2}) by $S_A \times S_B \times S_C$.  Denoting $s_{AB} \in S_{AB}$, $s_{AC} \in S_{AC}$, $s_{BC} \in S_{BC}$, downstairs one gets a set of cosets $\lbrace \, s_{AB} s_{AC} s_{BC} S_A S_B S_C \, \rbrace =
\lbrace \, s_{AB} s_{AC} s_{BC} S_A S_B S_A S_C S_B S_C \, \rbrace = \lbrace \, s_{AB} S_A S_B \, s_{AC} S_A S_C \, s_{BC} S_B S_C \, \rbrace$ which is $E_{AB} \cdot E_{AC} \cdot E_{BC}$.}

%%%%%%%%%%%%%%%%%%%%%%%%
\subsection{Examples\label{sect:examples}}
%%%%%%%%%%%%%%%%%%%%%%%%
Here we give a few examples to illustrate the approach, using a variety of three-qubit states.  The continuous stabilizer groups for these states were already
classified in \cite{Carteret_2000}.  However the definition of entanglement in terms of entanglement groups leads to a characterization of two-party and
three-party entanglement for these states which differs from \cite{Carteret_2000}.  We discuss these differences in detail in section
\ref{sect:usualnotion}.  For now we just present a few examples of stabilizer and entanglement groups.

\noindent{\underline{Tensor product state}} \\
We start with the tensor product state $\ket{\psi} = \ket{000}$.
One-party stabilizers $S_A,S_B,S_C$ correspond to unitary matrices
\be
u_A = e^{i \theta_A} \left(\begin{array}{ll} 1 & 0 \\ 0 & e^{i\phi_A} \end{array}\right)
\qquad
u_B = e^{i \theta_B} \left(\begin{array}{ll} 1 & 0 \\ 0 & e^{i\phi_B} \end{array}\right)
\qquad
u_C = e^{i \theta_C} \left(\begin{array}{ll} 1 & 0 \\ 0 & e^{i\phi_C} \end{array}\right)
\ee
There are no non-trivial two- or three-party stabilizers, so
$S_{ABC} = S_A \times S_B \times S_C$.  This means that the full entanglement group $F_{ABC}$ is trivial.

\noindent{\underline{Generalized GHZ}} \\
Consider the 3-qubit state
\begin{equation}
\ket{\psi}=a|000 \rangle+b|111 \rangle
\label{sghz}
\end{equation}
known as the generalized GHZ state.
For generic $a,b$ the full stabilizer group corresponds to matrices
\begin{equation}
\label{ghzfull}
S_{ABC} \, : \quad e^{i (\theta_A + \theta_B + \theta_C)}
\begin{pmatrix}
1 & 0 \\
0 & e^{i\phi_A}
\end{pmatrix}
\otimes
\begin{pmatrix}
1 & 0 \\
0 & e^{i\phi_B}
\end{pmatrix}
\otimes
\begin{pmatrix}
1 & 0 \\
0 & e^{-i(\phi_A + \phi_B)}
\end{pmatrix}
\end{equation}
From this we identify the two-party stabilizer subgroups as
\bea
\label{ghz2p1}
&& S_{AB} \, : \quad
e^{i(\theta_A + \theta_B)}
\begin{pmatrix}
1 & 0 \\
0 & e^{i\phi}
\end{pmatrix}
\otimes
\begin{pmatrix}
1 & 0 \\
0 & e^{-i\phi}
\end{pmatrix}
\otimes
\begin{pmatrix}
1 & 0 \\
0 & 1
\end{pmatrix} \\[3pt]
\label{ghz2p2}
&&S_{AC} \, : \quad e^{i(\theta_A + \theta_C)}
\begin{pmatrix}
1 & 0 \\
0 & e^{i\phi}
\end{pmatrix}
\otimes
\begin{pmatrix}
1 & 0 \\
0 & 1
\end{pmatrix}
\otimes
\begin{pmatrix}
1 & 0 \\
0 & e^{-i\phi}
\end{pmatrix} \\[3pt]
\label{ghz2p3}
&& S_{BC} \, : \quad e^{i(\theta_B + \theta_C)}
\begin{pmatrix}
1 & 0 \\
0 & 1
\end{pmatrix}
\otimes
\begin{pmatrix}
1 & 0 \\
0 & e^{i\phi}
\end{pmatrix}
\otimes
\begin{pmatrix}
1 & 0 \\
0 & e^{-i\phi}
\end{pmatrix}
\eea
The 1-party stabilizers are just phases times the identity operator, so the two-party entanglement groups $E_{AB}, E_{AC}, E_{BC}$ are
\be
\label{ghztwo}
E_{AB} \approx E_{AC} \approx E_{BC} \approx U(1)
\ee
However for generic $a,b$ the quotient
\begin{equation}
E_{ABC} = \frac{S_{ABC}}{S_{AB}\cdot S_{BC}\cdot S_{AC}} = \lbrace \, 1 \, \rbrace
\end{equation}
is trivial.

\noindent{\underline{Standard GHZ}} \\
If one sets $a=b=1/\sqrt{2}$ in (\ref{sghz}) to recover the standard GHZ state
\be
\label{GHZstate}
\ket{\psi} = {1 \over \sqrt{2}} \left( \ket{000} + \ket{111} \right)
\ee
the stabilizer is enlarged.  One gains an additional stabilizer
\be
\label{GHZuAuBuC}
x = \begin{pmatrix}
0 & 1 \\
1 & 0
\end{pmatrix}
\otimes
\begin{pmatrix}
0 & 1 \\
1 & 0
\end{pmatrix}
\otimes
\begin{pmatrix}
0 & 1 \\
1 & 0
\end{pmatrix} \in S_{ABC}
\ee
The additional stabilizer generates a ${\mathbb Z}_2$ subgroup of $S_{ABC}$ since $x^2 = 1$.  Denoting the stabilizers appearing in (\ref{ghzfull}) by $s_{\theta\phi}$, the
full stabilizer group for standard GHZ becomes
\be
S_{ABC} = \lbrace \, s_{\theta\phi},\, x s_{\theta\phi} \, \rbrace \quad \hbox{\rm where $s_{\theta\phi} \in \hbox{\rm (stabilizer group for generalized GHZ)}$}
\ee
The two-party stabilizer groups are unchanged, so for standard GHZ the two-party entanglement groups are still given in (\ref{ghztwo}).  However the
three-party entanglement group is enlarged and becomes non-trivial.
\begin{equation}
E_{ABC} = \frac{S_{ABC}}{S_{AB}\cdot S_{AC}\cdot S_{BC}} = {\mathbb Z}_2
\end{equation}

\noindent\underline{The W state} \\
Consider the 3-qubit state\footnote{The standard W state is obtained by setting $a = c = d = 1/\sqrt{3}$ and acting
with $\left(\begin{array}{cc} 0 & 1 \\ 1 & 0 \end{array}\right) \otimes I \otimes I$.  We present it in this form for compatibility with (\ref{generic3}).}
\begin{equation}
\ket{\psi} = a \ket{000} +c \ket{101} + d \ket{110}
\end{equation}
known as the W state.  This state is stabilized up to a phase by matrices of the form
\begin{equation}
\label{WuAuBuC}
e^{i(\theta_A + \theta_B + \theta_C)}
\begin{pmatrix}
e^{i\phi}& 0 \\
0 &e^{-i\phi}
\end{pmatrix}
\otimes
\begin{pmatrix}
e^{-i\phi}& 0 \\
0 & e^{i\phi}
\end{pmatrix}
\otimes
\begin{pmatrix}
e^{-i\phi} & 0 \\
0 & e^{i\phi}
\end{pmatrix}
\end{equation}
The one-party and two-party stabilizers are just phases times the identity operator, so $E_{AB}, E_{AC}, E_{BC}$ are trivial.  However there
is three-party entanglement, with $E_{ABC} = U(1)$.

\noindent{\underline{The ace state}} \\
We now consider
 \begin{equation}
 |\psi_{ace}\rangle = a \ket{000}+c \ket{101}+e \ket{111} 
  \end{equation}
 For generic parameters it has a stabilizer
\begin{equation}
S_{ABC} \, : \quad e^{i\theta}
\begin{pmatrix}
e^{i\phi}& 0 \\
0 &e^{-i\phi}
\end{pmatrix}
\otimes
\begin{pmatrix}
1& 0 \\
0 & 1
\end{pmatrix}
\otimes
\begin{pmatrix}
e^{-i\phi} & 0 \\
0 & e^{i\phi}
\end{pmatrix}.
\end{equation}
There is no three-party entanglement in the generic state, $E_{ABC} = \lbrace 1 \rbrace$.  There is two-party entanglement between $A$ and $C$, with $E_{AC} \approx U(1)$.  However there is no entanglement between $A$ and $B$ or between
$B$ and $C$ since $E_{AB}$ and $E_{BC}$ are both trivial.  If we group qubits $A$ and $C$ together to form a combined system $(AC)$, then as in section \ref{sect:classifying} qubit $B$ is
entangled with the combined system in the sense that $E_{(AC)B}$ is non-trivial.

For non-generic parameters the state can have more entanglement. For instance if $a^2 = c^2 + e^2$ there is an additional discrete stabilizer
\begin{equation}
\label{ace3party}
\begin{pmatrix}
0& 1 \\
1 &0
\end{pmatrix}
\otimes
\begin{pmatrix}
{c / a}& {e / a} \\[3pt]
\, {e / a} & - {c / a}
\end{pmatrix}
\otimes
\begin{pmatrix}
0 & 1 \\
1 &0
\end{pmatrix}
\end{equation}
which squares to the identity.  In this case the state does have three-party entanglement, with $E_{ABC} \approx \mathbb{Z}_{2}$.

\noindent{\underline{Generic state}} \\
Finally we consider a generic state of three qubits.  Using local unitary transformations the most general state of three qubits can be put in the form
\cite{PhysRevLett.85.1560}
\begin{equation}
\label{generic3}
\ket{\psi} = a \ket{000} +b e^{i\phi} \ket{100} +c \ket{101} + d \ket{110} +e \ket{111}
\end{equation}
with $a,b,c,d,e,\phi$ real. 
For generic values of the parameters the stabilizer groups just consist of phases times the
identity operator, so all of the
entanglement groups are trivial.  In this
sense the generic state has no entanglement relative to the partition into three separate systems $A$, $B$, $C$, even though the generic state is
not a tensor product.

Let us emphasize that we are saying there is no entanglement {\em relative to the partition into three
separate systems}.  This means there is no unitary operator $u_A$ on system $A$ that can be undone by a factorized unitary operator
$u_B \otimes u_C$ that acts on systems $B$ and $C$.  However, as explained in \ref{sect:classifying}, the entanglement of a state is defined by considering all possible partitions.
We should therefore consider different partitions of the system.  For example
we could look for entanglement between system $A$ and the combined $(BC)$ system.  This means looking for a unitary operator $u_A$
that acts on system $A$ and can be undone by a unitary operator $u_{(BC)}$ that acts on the combined $(BC)$ system.
Such operators generically do exist, and in this sense the generic state does have entanglement.

%%%%%%%%%%%%%%%%%%%%%%%
\subsection{Multi-partite systems\label{sect:multipartite}}
%%%%%%%%%%%%%%%%%%%%%%%
The generalization to $N$-partite systems is straightforward.  We begin with a pure state $\ket{\psi}$ in a Hilbert space ${\cal H} = {\cal H}_1 \otimes
{\cal H}_2 \otimes \cdots \otimes {\cal H}_N$ where the factors have dimension $d_1,d_2,\ldots,d_N$.  The full stabilizer is the subgroup
\be
S_{12 \cdots N} \subset U(d_1) \times U(d_2) \times \cdots \times U(d_N)
\ee
that leaves the state invariant up to a phase.  This has normal subgroups which act trivially on some of the systems.  We have
\bea
\nonumber
&& \hbox{\rm $N$ one-party stabilizers} \qquad\,\,\,\,\,\, S_i \, \qquad i = 1,\ldots N \\
&& \hbox{\rm $\left({N \atop 2}\right)$ two-party stabilizers} \qquad S_{ij} \qquad 1 \leq i < j \leq N \\
\nonumber
&& \hbox{\rm $\left({N \atop 3}\right)$ three-party stabilizers} \quad\,\, S_{ijk} \qquad 1 \leq i < j < k \leq N
\eea
and so on.

The full entanglement group, which captures all entanglement relative to the partition into $N$ systems, is defined by
\be
F_{1 2 \cdots N} = S_{1 2 \cdots N} / (S_1 \times S_2 \times \cdots \times S_N)
\ee
We can also characterize entanglement between collections of subsystems.  Choose $n$ subsystems
\be
i_1 < i_2 < \cdots < i_n
\ee
$n$-party entanglement between these systems, which cannot be built up from entanglement between fewer than $n$ systems, is characterized
by the group
\be
E_{i_1 i_2 i_3 \cdots i_n} = S_{i_1 i_2 i_3 \cdots i_n} / (S_{i_2 i_3 \cdots i_n} \cdot S_{i_1 i_3 \cdots i_n} \cdot \ldots \cdot S_{i_1 i_2 \cdots i_{n-1}})
\ee
where we quotient the appropriate $n$-party stabilizer by the group product of all the $(n-1)$-party stabilizers that it contains.

%%%%%%%%%%%%%%%%%%%%%%%%%%%%
\subsection{Physical meaning of entanglement groups\label{sect:physical}}
%%%%%%%%%%%%%%%%%%%%%%%%%%%%
What physical properties does an entanglement group capture?  As we introduced it, the starting point is a stabilizer group, which consists
of local unitary operations on some parts of a system that are equivalent to, or can be undone by, local unitary operations on other parts.
This provides an operational meaning to entanglement, by relating physical processes that are carried out in different laboratories.

Another perspective is that local unitary stabilizers capture correlations between different parts of a system.  A local unitary stabilizer $s$ can be written in terms of
local Hermitian generators $H_i$ as
\be
s = e^{i \sum_i H_i}
\ee
Here $H_i$ only acts on subsystem $i$, and if the stabilizer doesn't act on some of the subsystems, it just means the corresponding $H_i = 0$.
Since $s \vert \psi \rangle = e^{i \theta} \vert \psi \rangle$, we see that measurements of the observables $H_i$, leading to outcomes $h_i$, will yield
correlated results according to\footnote{The eigenvalues $h_i$, as well as the angle $\theta$, are only defined modulo $2 \pi$.  In writing (\ref{correlation})
we're implicitly using this freedom to fix say $0 \leq \sum_i h_i < 2\pi$ and $0 \leq \theta < 2 \pi$.  More generally, without using this freedom, we could
write $\sum_i h_i \equiv \theta \, {\rm mod} \, 2 \pi$.}
\be
\label{correlation}
\sum_i h_i = \theta
\ee
In this way, a stabilizer implies that individual measurements performed on different parts of the system are perfectly correlated.  From this perspective, $n$-party
$g$-entanglement is responsible for those correlations which cannot be explained in terms of correlations involving $n-1$ or fewer systems.
The full entanglement group encodes the mathematical  structure of all such correlations.
These correlations have implications when trying to accomplish physical tasks using entangled states, as described in section \ref{section:qtasks}.

We want to stress that, even though very few states display $g$-entanglement when divided into $k>2$ subsystems \cite{10.1063/1.5003015,PhysRevX.8.031020}, the physical properties of a state
are often controlled by how close they are to a state with an enhanced stabilizer.  Suppose a particular state $\vert \psi \rangle$
has a $k$-party stabilizer that enforces the correlation (\ref{correlation}) between measurements.  States close to $\vert \psi \rangle$ may not have a
$k$-party stabilizer at all.  However they will have a probability distribution for measuring $\sum_i h_i$ that is sharply peaked around
$\theta$.  In this way the existence of nearby states with an enhanced stabilizer can drive physical behavior.

As an example of this reasoning, states close to GHZ may win the GHZ game with a probability that beats the
classical strategy, and likewise may succeed in teleportation or other quantum tasks, but this is due to the fact that they are close to a state that has probability 1 for success.
In a similar vein, states without a stabilizer can violate Bell-like inequalities, but maximum violation is generally associated with states that have an
enhanced stabilizer \cite{G_hne_2005}.

This is similar to other situations in physics, in which states with an enhanced symmetry govern the behavior in a neighborhood of the enhanced
symmetry point.  This happens near a second-order phase transition in a condensed matter system, or near a symmetry-enhanced point in quantum
field theory.  We will see an example of this in the context of the GHZ game in section \ref{subsection:qpt}.

%%%%%%%%%%%%%%%%%%%%%%%%%%%
\subsection{Classifying the entanglement of a state\label{sect:classifying}}
%%%%%%%%%%%%%%%%%%%%%%%%%%%
So far we have considered entanglement of a state relative to a fixed partition into subsystems $A,B,C,\ldots$.  However it is possible that
any given partition will not be sensitive to all of the entanglement that is present in the state.  To fully explore the entanglement of a state
we have to consider all possible partitions into subsystems.

For example, in section \ref{sect:tripartite} we considered a system which was divided into three parts $A,B,C$.  However we could group
$A$ and $B$ together and regard the Hilbert space as
\be
{\cal H} = \left({\cal H}_A \otimes {\cal H}_B\right) \otimes {\cal H}_C
\ee
We will use the notation $(AB)C$ to indicate that we are grouping systems $A$ and $B$ together in this way.
This grouping doesn't change the Hilbert space, but it enlarges the group of local unitary transformations from
$U(d_A) \times U(d_B) \times U(d_C)$ to $U(d_A d_B) \times U(d_C)$.  This enlarged group gives more possibilities for entanglement,
since we can consider operators of the form $u_{(AB)} \otimes u_C$, in which the unitary $u_{AB}$ that acts on the combined $(AB)$
system does not factor as $u_A \otimes u_B$.
We can look for stabilizers in the enlarged local unitary group, which we will denote
\be
S_{(AB)C} \subset U(d_A d_B) \times U(d_C)
\ee
and we can construct an entanglement group
\be
E_{(AB)C} = S_{(AB)C} / \big(S_{(AB)} \times S_C\big)
\ee
With the grouping $(AB)C$, we have effectively made a bipartite division of the system, so the definition of $E_{(AB)C}$ is modeled on the
discussion of bipartite entanglement in section \ref{sect:bipartite}.  We could of course consider multipartite entanglement groups for any
way of partitioning the system.

Let us illustrate this with an example.  For the generic 3-qubit state (\ref{generic3}) we saw that the entanglement groups were all trivial when
the system was partitioned into three distinct spins.  However we can group spins $A$ and $B$ together, and ask for the entanglement
group $E_{(AB)C}$ of a generic state.  This is a straightforward exercise since, as discussed in appendix \ref{appendix:bipartite}, we can write
the generic state (\ref{generic3}) in a Schmidt basis as
\be
\ket{\psi} = \sum_{i = 1}^2 p_i \, \ket{i}_{AB} \otimes \ket{i}_C
\ee
Here $\ket{i}_{AB}$ and $\ket{i}_C$ are orthonormal bases for ${\cal H}_A \otimes {\cal H}_B$ and ${\cal H}_C$, and for a generic state
$p_1 > p_2 > 0$.  As discussed in the appendix the generic entanglement group is $E_{(AB)C} = U(1)$.  The other possible partitions
generically have the same behavior, with $E_{(AC)B} = E_{(BC)A} = U(1)$.  This procedure captures all of the entanglement present
in the generic three-qubit state.  Note that unless the original state was a pure tensor product $\ket{\psi_A} \otimes \ket{\psi_B} \otimes \ket{\psi_C}$
then at least one of the entanglement groups $E_{(AB)C}$, $E_{(AC)B}$, $E_{(BC)A}$ will be non-trivial.

We can now give a prescription for classifying the entanglement of a state.  Suppose a system consists of $N$ elementary spins.
We partition the spins in all possible ways into disjoint sets $A_1,A_2,\ldots,A_k$ and we compute the entanglement group
$E_{(A_1)(A_2)\cdots(A_k)}$ associated with each partition.  This gives a large collection of entanglement groups associated with the state.
We say that two states have the same type or class of $g$-entanglement if this procedure leads to the same collection of entanglement groups.

This classification scheme has a few properties.
\begin{itemize}
\item Two states of an $N$-particle system $|\psi_{1} \rangle$ and $|\psi_{2} \rangle$ which are related by a local unitary transformation have the same entanglement type.
\item Since bipartite entanglement for a state (i.e.\ when there are only two sets $A_1$ and $A_2$) is trivial only for a tensor product, every state other than the pure tensor
product $\ket{\psi_1} \otimes \cdots \otimes \ket{\psi_N}$ will be in some non-trivial entanglement class.
\item  For finite $N$ and with a finite-dimensional Hilbert space for each particle, the classification is finite. That is, there are only a finite number of types of entanglement.
A state can jump from one type to another as the state is varied in a continuous manner.
\end{itemize}

%%%%%%%%%%%%%%%%%%%%%%%%%%%%%%%%%%%%%%%%%%%
\subsection{Comparing to the usual notion of entanglement\label{sect:usualnotion}}
%%%%%%%%%%%%%%%%%%%%%%%%%%%%%%%%%%%%%%%%%%%
The notion of $g$-entanglement we have introduced, based on groups, differs in some important respects from the usual notion of $u$-entanglement
which is based on tensor products and separability.  Here we discuss the differences and illustrate with some examples.

We begin with the simple case of a completely factorized tensor product state $\vert \psi \rangle_A \otimes \vert \psi \rangle_B \otimes \cdots$.
In this case the two notions agree.  The entanglement groups are all trivial, no matter how the subsystems $A,\, B,\, \ldots$ are grouped together, so there
is no $g$-entanglement present in a tensor product state.  At the same time a tensor product state has no entanglement according to the usual notion.
So according to both $g$-entanglement and $u$-entanglement, a tensor product state has no entanglement.

As another instructive example, an $N$-partite state is usually called fully entangled if it cannot be written as a tensor
product for any grouping of its parts.  This notion of being fully entangled is captured by the entanglement groups.  It corresponds to the statement that in a fully entangled state, the
entanglement groups $E_{(A)(\bar{A})}$ are non-trivial for every $A$.  (Here $A$ denotes a collection of subsystems, grouped into a single object, and $\bar{A}$ denotes the complement of $A$.
Both $A$ and $\bar{A}$ must contain at least one subsystem.)  To see this, note that $E_{(A)(\bar{A})}$ is effectively a bipartite entanglement group, and as shown in appendix \ref{appendix:bipartite},
$E_{(A)(\bar{A})}$ is trivial if and only if the state is a tensor product $\vert \psi \rangle_A \otimes \vert \psi \rangle_{\bar{A}}$.

So far we've seen that a fully-entangled state is characterized by a non-trivial $E_{(A)(\bar{A})}$ for every bipartite division into $A$ and $\bar{A}$.  However the $g$-entanglement scheme provides
us with a more refined notion of how entanglement is distributed among subsystems, since we can ask about the multi-party entanglement group associated with any list of subsystems.
For example, in the $g$-entanglement scheme, $N$-partite entanglement is characterized by a non-trivial $N$-party entanglement group $E_{A_1 \ldots A_N}$.  Even if a state is fully entangled, it may have no $N$-party stabilizers, which means the $N$-party entanglement
group may be trivial.  As pointed out in section \ref{sect:physical}, even though a particular state may not have a multi-party stabilizer,
the existence of a nearby state with an enhanced stabilizer can be an important driver of physical behavior.

Finally we consider $k$-partite entanglement.  The usual notion of whether there is $k$-partite entanglement in an $N$-partite system
relies on whether the $k$-partite reduced density matrix $\rho_{A_1 \ldots A_k}$ is separable.  Our definition uses a different criterion, namely, whether the entanglement group
$E_{A_1 \ldots A_k}$ is trivial.  One might inquire about the origin of the difference.  Shouldn't a (fully) separable $k$-party reduced density
matrix imply that the $k$-party entanglement group is trivial?  As we will show in section \ref{sect:separable}, this is not quite the case.  Even though a reduced density matrix $\rho_{A_1 \ldots A_k}$
is separable, the entanglement group $E_{A_1 \ldots A_k}$ may be non-trivial.  However the stabilizers that contribute to $E_{A_1 \ldots A_k}$ take a very restricted form: they can be
written as a product of two other stabilizers that involve the complement of $A_1 \ldots A_k$.\footnote{See section \ref{subsec:EntGpForSeparable} and appendix \ref{appendix:twoparty}.  This also requires
that the stabilizers correspond to central elements of the entanglement group.}  In this way the entanglement group
$E_{A_1 \ldots A_k}$ can be non-trivial, even if the reduced density matrix $\rho_{A_1 \ldots A_k}$ is separable, however the entanglement
can be thought of as built up indirectly, from entanglement with the purifying system.

We conclude with a few examples to illustrate the differences between $g$-entanglement and $u$-entanglement, building on the results in section
\ref{sect:examples}.

\noindent{\underline{Generalized GHZ}} \\
The two-party entanglement groups $E_{AB}, E_{AC}, E_{BC}$ are non-trivial, so there is two-party $g$-entanglement between any pair of qubits.
This happens even though the two-party reduced density matrices are separable, which means there is no two-party $u$-entanglement.  What is the
origin of the two-party $g$-entanglement?  As discussed above, it comes from a situation in which an $AC$ stabilizer can be
built as a product of an $AB$ stabilizer with a $BC$ stabilizer.\footnote{This can be seen explicitly in (\ref{ghz2p2}), where $S_{AB}$ times $S_{BC}$ gives $S_{AC}$.}
There is a certain ambiguity, whether one should think of such elements of $E_{AC}$ as reflecting $AC$ entanglement or a combination of $AB$ and $BC$
entanglement.  The entanglement groups treat this situation democratically.

For the generalized GHZ state the three-party entanglement group $E_{ABC}$ is trivial, so there is no three-party $g$-entanglement, even though
the state is fully entangled according to the usual notion.  The point is that the generalized GHZ state is not bi-separable, and has non-trivial entanglement
groups $E_{A(BC)}$, $E_{B(AC)}$, $E_{C(AB)}$, however the group $E_{ABC}$ is trivial.  This is true even though the three-tangle is non-zero,
which suggests that (as seen in other contexts) the three-tangle does not provide a comprehensive description of three-party entanglement.

\noindent{\underline{Standard GHZ}} \\
The discussion of two-party $g$-entanglement proceeds just as for generalized GHZ.  However the standard GHZ state has a non-trivial 3-party entanglement group,
$E_{ABC} = {\amsbb Z}_2$.  One can view this enhanced entanglement group as providing an explanation for why the three-tangle, which depends
continuously on the parameters of the state, is non-zero even in the vicinity of standard GHZ.

\noindent{\underline{W state}} \\
The two-party entanglement groups are trivial, so there is no two-party $g$-entanglement in the W state.  However the reduced two-party density matrices are not separable,
which means there is two-party $u$-entanglement.  How does this come about?  The reduced density matrices for the W state are not separable, and the W state indeed has
non-trivial entanglement groups $E_{A(BC)}$, $E_{B(AC)}$, $E_{C(AB)}$ as well as $E_{ABC}$.  However the groups $E_{AB}$, $E_{AC}$, $E_{BC}$ are trivial.  In our scheme
one could say that $A$ and $B$ are entangled, but in a way that necessarily involves system $C$, so it does not count as two-party $g$-entanglement between $A$ and $B$.
Note that the three-tangle is zero, even though the W state has three-party $g$-entanglement.

%%%%%%%%%%%%%%%%%%%%%%%%%%%%%%%%
\section{Structure of entanglement and monogamy\label{sect:monogamy}}
%%%%%%%%%%%%%%%%%%%%%%%%%%%%%%%%
We'd like to understand the structure of the entanglement groups, their relationship to one another and how they act on individual systems.  To this end
we introduce projection operators $\pi_A,\,\pi_B,\,\pi_C$ onto systems $A,B,C$.  For example if $g = (g_A,g_B,g_C)$
then $\pi_A(g) = g_A$.

Suppose a group element $g_A$ belongs to a particular entanglement group, say $g_A \in \pi_A(E_{AB})$.  Here are some questions we'd
like to understand.
\begin{enumerate}
\item
Is there a unique two-party entanglement transformation $g_{AB} \in E_{AB}$ that projects to give $g_A$?
If so, can we say how $g_{AB}$ acts on system $B$?
\item
Could the same group element appear in other entanglement groups?  That is, could $g_A$ also be an element of $\pi_A(E_{AC})$?
\end{enumerate}
Below we will establish
\begin{enumerate}
\item
An isomorphism theorem: there is a unique transformation $g_{AB} \in E_{AB}$ that projects to give $g_A$.  It acts isomorphically on systems
$A$ and $B$.
\item
A no-sharing theorem: $g_A$ can appear in more than entanglement group, but only if it is an element of the center of both $\pi_A(E_{AB})$ and
$\pi_A(E_{AC})$.  Otherwise it cannot belong to more than one entanglement group.
\end{enumerate}
Although our description has been for tripartite systems, these statements have generalizations to multipartite systems.  They provide a sharp statement of monogamy of
entanglement in terms of group theory.

%%%%%%%%%%%%%%%%%%%%%%%%%%%%%%
\subsection{Isomorphism theorems\label{sect:isomorphism}}
%%%%%%%%%%%%%%%%%%%%%%%%%%%%%%
We begin with a discussion of two-party entanglement, say between systems $A$ and $B$.  Other systems $C,D,\ldots$ could be present
but are not relevant.  Two-party entanglement is characterized by
\be
E_{AB} = S_{AB} / (S_A \times S_B)
\ee
From the point of view of system $A$ we are starting from a group of transformations $\pi_A(S_{AB})$ that can be undone by a compensating
transformation on system $B$.  However we consider $\pi_A(S_A)$, the subgroup of transformations on $A$ that don't require a
compensating transformation on $B$, to be trivial.  In this way entanglement transformations on system $A$ are characterized by
$\pi_A(S_{AB}) / \pi_A(S_A)$.  A similar perspective is available from the point of view of system $B$ and leads to $\pi_B(S_{AB}) / \pi_B(S_B)$.

These two perspectives are related thanks to Goursat's lemma.  As reviewed in appendix \ref{appendix:Goursat} there are isomorphisms
\be
\label{IsomorphismReference}
E_{AB} \approx \pi_A(S_{AB}) / \pi_A(S_A) \approx \pi_B(S_{AB}) / \pi_B(S_B)
\ee
That is, elements of $E_{AB}$ have the form
\be
\label{diagonal}
(g_{A},\phi(g_{A}))
\ee
where $\phi \, : \, \pi_A(S_{AB}) / \pi_A(S_A) \rightarrow \pi_B(S_{AB}) / \pi_B(S_B)$ is an isomorphism.

Note that for a bipartite system, where systems $A$ and $B$ are all there is, the isomorphism can be understood from the Schmidt decomposition (\ref{matrix}), which shows that
elements of the entanglement group have the form $u \otimes u^* \in E_{AB}$.  This can also be seen explicitly in the 2-qubit
examples (\ref{genericuAuB}), (\ref{BelluAuB}).

Next we consider three-party entanglement.  For concreteness we consider a tripartite system.  We'd like to understand how three-party entanglement acts on the individual systems $A$, $B$, $C$.
To save on writing we'll denote the three-party entanglement group by $E_{ABC} = G / N$ where $G = F_{ABC}$ and $N = E_{AB} \cdot E_{AC} \cdot E_{BC}$.
Also we'll denote the projections on the various factors by $G_A = \pi_A(G)$, $N_A = \pi_A(N)$, etc.
 
 As we show in appendix \ref{appendix:Goursat} an extension of Goursat's lemma applies to this situation and implies that $E_{ABC}$ acts isomorphically on the three factors.
 \be
 E_{ABC} \approx G_A / N_A \approx G_B / N_B \approx G_C / N_C
 \ee
 Moreover $E_{ABC}$ is diagonally embedded in the direct product $(G_A / N_A) \times (G_B / N_B) \times (G_C / N_C)$ in the sense that there are
 isomorphisms
 \bea
 \nonumber
&& \theta_A \, : \, E_{ABC} \rightarrow G_A / N_A \\
&& \theta_B \, : \, E_{ABC} \rightarrow G_B / N_B \\
\nonumber
&& \theta_C \, : \, E_{ABC} \rightarrow G_C / N_C
\eea
with
\be
E_{ABC} \approx \lbrace \, \big(\theta_A(t),\,\theta_B(t),\, \theta_C(t) \big) \, \vert \, t \in E_{ABC} \, \rbrace
\ee
So three-party entanglement transformations act isomorphically on systems $A$, $B$ and $C$.  This can be seen explicitly for the GHZ
and W states in (\ref{GHZuAuBuC}), (\ref{WuAuBuC}).  It can also be seen for the non-generic ace state in (\ref{ace3party}), since each
matrix appearing in (\ref{ace3party}) generates a ${\mathbb Z}_2$ on the corresponding system.

We expect this type of relation to generalize: $n$-partite entanglement transformations should act isomorphically on the
$n$ systems that participate in the transformation.

%%%%%%%%%%%%%%%%%%%%%%%%%%%%
\subsection{No-sharing theorems\label{sect:noshare}}
%%%%%%%%%%%%%%%%%%%%%%%%%%%%
Consider a transformation $g_A \in \pi_A(E_{AB})$.  Could this same transformation appear in the projection of any other two-party
entanglement group?  The answer is yes, but with restrictions.  To understand the restrictions we first return to consider the
one- and two-party stabilizer groups in a little more detail.

The one-party stabilizers $S_A$, $S_B$, $S_C$ have trivial intersection and commute with each other.
However the two-party stabilizer groups $S_{AB}$, $S_{AC}$, $S_{BC}$ can overlap and need not commute.  Instead of trivial overlap, the intersection of the groups obey 
\be
\label{intersect}
S_{AB} \cap S_{AC} = S_A
\ee
In addition if $s_{AB} \in S_{AB}$ and $s_{AC} \in S_{AC}$ then their group commutator can give a one-party stabilizer.
\be
\label{st}
[s_{AB},s_{AC}] \equiv s_{AB} s_{AC} s_{AB}^{-1} s_{AC}^{-1} \in S_A
\ee

What does this imply for the two-party entanglement groups?  The commutator (\ref{st}) becomes trivial when we quotient by $S_A$, which means $E_{AB}$ and $E_{AC}$ are subgroups of
$F_{ABC}$ that commute element-by-element.
\be
\label{e}
e_{AB} e_{AC} = e_{AC} e_{AB} \qquad \hbox{\rm$\forall$ $e_{AB} \in E_{AB}$, $e_{AC} \in E_{AC}$}
\ee
When projected on system $A$ this implies
\be
\pi_{A}(e_{AB}) \pi_{A}(e_{AC}) = \pi_{A}(e_{AC})\pi_{A} (e_{AB})
\label{eecom}
\ee

To return to our original question, suppose we have a transformation $g_A$ which is an element of both $\pi_A(E_{AB})$ and
$\pi_A(E_{AC})$.  Since different two-party entanglement groups commute element-by-element (see (\ref{eecom})), $g_A$ must be an element of
the center of both groups.
\be
\label{bipartiteshare}
g_A \in Z(\pi_A(E_{AB})) \cap Z(\pi_A(E_{AC}))
\ee
It is possible for centers to intersect, so this is not a vacuous possibility.  This can be seen explicitly for the generalized GHZ state in (\ref{ghz2p1}),
(\ref{ghz2p2}) which gives $\pi_A(E_{AB}) = \pi_A(E_{AC}) = U(1)$.

This lets us make statements about the monogamy of two-party entanglement.
\begin{itemize}
\item
Elements $g_A \in \pi_A(E_{AB})$ that are not in the common center $Z(\pi_A(E_{AB})) \cap Z(\pi_A(E_{AC}))$ are monogamous.
They reflect entanglement between $A$ and $B$ but cannot appear in any other two-party entanglement group.
\item
Elements $g_A \in \pi_A(E_{AB})$ that do appear in the common center $Z(\pi_A(E_{AB})) \cap Z(\pi_A(E_{AC}))$ are not
monogamous.  They are shared between two different two-party entanglement groups.
\end{itemize}
We see that strictly speaking entanglement is not monogamous in the sense that the common center can be shared between two
different entanglement groups.  Note that common center is, by definition, an Abelian group and that due to the isomorphic embedding
(\ref{diagonal}) it acts the same way on systems $A,B,C$.

The fact that different entanglement groups can share a common center leads to a certain freedom in the way in which entanglement
transformations are presented.  Let us illustrate the phenomenon for a tripartite system.
Thanks to (\ref{intersect}) distinct two-party entanglement groups have no non-trivial elements in common.
\be
S_{AB} \cap S_{AC} = S_A \quad \Rightarrow \quad E_{AB} \cap E_{AC} = \lbrace \, 1 \, \rbrace
\ee
However suppose the projected centers have elements in common and let
\be
g_A \in Z(\pi_A(E_{AB})) \cap Z(\pi_A(E_{AC}))
\ee
This means there are elements $(g_A,g_B,1) \in E_{AB}$ and $(g_A,1,g_C) \in E_{AC}$.  But then there is an element of $E_{BC}$
which we can present as
\be
\label{ambiguous}
(1,g_B,g_C^{-1}) = (g_A,g_B,1) \cdot (g_A^{-1},1,g_C^{-1}) \in E_{BC}
\ee
This means a particular $BC$ entanglement cannot be distinguished from an $AB$ entanglement combined with an $AC$ entanglement.
Note that this possibility only exists if $g_A,g_B,g_C$ all belong to common projected centers.
As mentioned in footnote \ref{footnote} it leads to the possibility that $E_{BC} \cap (E_{AB} \cdot E_{AC}) \not= \lbrace 1 \rbrace$.

So far we have discussed monogamy for two-party entanglement but the considerations can be generalized.  Suppose we
have a multi-partite system which we divide as $A_1,\ldots,A_\ell,B_1,\ldots,B_m,C_1,\ldots,C_n$.  Stabilizer groups can intersect since
\be
S_{A_1 \cdots A_\ell B_1 \cdots B_m} \cap S_{A_1 \cdots A_\ell C_1 \cdots C_n} = S_{A_1 \cdots A_\ell}
\ee
Given elements $s_{A_1 \cdots A_\ell B_1 \cdots B_m} \in S_{A_1 \cdots A_\ell B_1 \cdots B_m}$ and
$s_{A_1 \cdots A_\ell C_1 \cdots C_n} \in S_{A_1 \cdots A_\ell C_1 \cdots C_n}$
their group commutator (defined in (\ref{st})) can give another stabilizer.
\be
\label{multisABsAC}
[s_{A_1 \cdots A_\ell B_1 \cdots B_m},s_{A_1 \cdots A_\ell C_1 \cdots C_n}] \in S_{A_1 \cdots A_\ell}
\ee
Now suppose we have a transformation
\be
g_A \in \pi_{A_1 \cdots A_\ell}(S_{A_1 \cdots A_\ell B_1 \cdots B_m}) \cap \pi_{A_1 \cdots A_\ell}(S_{A_1 \cdots A_\ell C_1 \cdots C_n})
\ee
We can think of such a transformation as coming from an element of $S_{A_1 \cdots A_\ell C_1 \cdots C_n}$, which by (\ref{multisABsAC})
means its commutator with any other element of $\pi_{A_1 \cdots A_\ell}(S_{A_1 \cdots A_\ell B_1 \cdots B_m})$ will be an element of
$S_{A_1 \cdots A_\ell}$.  This commutator becomes trivial if we quotient by $S_{A_1 \cdots A_\ell}$, so
\be
[g_A] \in Z[\pi_{A_1 \cdots A_\ell}(S_{A_1 \cdots A_\ell B_1 \cdots B_m}) / S_{A_1 \cdots A_\ell}]
\ee
The entanglement group $E_{A_1 \cdots A_\ell B_1 \cdots B_m}$ can be obtained by taking further quotients.  But
an element of the center projects to an element of the center under additional quotients,\footnote{If $g$ is in the center, so that $f g = g f$ for all $f \in G$,
then after a quotient by $H \, \triangleleft \, G$ we have $[f] [g] = [fg] = [gf] = [g] [f]$ for all $[f] \in G / H$.} so
\be
[g_A] \in Z[\pi_{A_1 \cdots A_\ell}(E_{A_1 \cdots A_\ell B_1 \cdots B_m})]
\ee
Analogous reasoning gives
\be
[g_A] \in Z[\pi_{A_1 \cdots A_\ell}(E_{A_1 \cdots A_\ell C_1 \cdots C_n})]
\ee
This is the multipartite analog of (\ref{bipartiteshare}): an element $g_A$ that comes from the projection on $A_1 \cdots A_\ell$ of two different stabilizer groups
can appear in the corresponding projected entanglement groups, but only if it is an element of the center of both.

In the case of stabilizers which act on nested collections of subsystems we can make a stronger statement.  Consider stabilizers
$S_{A_1 \cdots A_\ell B_1 \cdots B_m}$ and $S_{A_1 \cdots A_\ell B_1 \cdots B_m C_1 \cdots C_n}$, and let
\be
g_A \in \pi_{A_1 \cdots A_\ell} (S_{A_1 \cdots A_\ell B_1 \cdots B_m})
\ee
Note that we automatically have $g_A \in \pi_{A_1 \cdots A_\ell} (S_{A_1 \cdots A_\ell B_1 \cdots B_m C_1 \cdots C_n})$.  
It is certainly possible for $g_A$ to participate in $AB$ entanglement, in the sense that the equivalence class of $g_A$ could be
a non-trivial element of $\pi_{A_1 \cdots A_\ell}(E_{A_1 \cdots A_\ell B_1 \cdots B_m})$.  Could $g_A$ also participate in $ABC$ entanglement, in the sense
of corresponding to a non-trivial equivalence class in $\pi_{A_1 \cdots A_\ell} ( E_{A_1 \cdots A_\ell B_1 \cdots B_m C_1 \cdots C_n})$?
The answer is no.  To see this imagine we have elements
\bea
\nonumber
&& g_1 = (g_A,g_B) \in S_{A_1 \cdots A_\ell B_1 \cdots B_m} \\
&& g_2 = (g_A,g_B',g_C') \in S_{A_1 \cdots A_\ell B_1 \cdots B_m C_1 \cdots C_n}
\eea
Then $g_1^{-1} g_2 = (1,g_B^{-1} g_B',g_C') \in S_{B_1 \cdots B_m C_1 \cdots C_n}$, which means $g_2$ can be written as a product of an element of $S_{A_1 \cdots A_\ell B_1 \cdots B_m}$
with an element of $S_{B_1 \cdots B_m C_1 \cdots C_n}$.
\be
g_2 = (g_A,g_B,1) \cdot (1,g_B^{-1} g_B',g_C')
\ee
Since $g_2$ can be written in this way it corresponds to the identity element of $E_{A_1 \cdots A_\ell B_1 \cdots B_m C_1 \cdots C_n}$.  Projecting on system $A$ this means
the equivalence class of $g_A$ is the identity element of $\pi_{A_1 \cdots A_\ell} ( E_{A_1 \cdots A_\ell B_1 \cdots B_m C_1 \cdots C_n})$.

%%%%%%%%%%%%%%%%%%%%%%%%%%%%%%%%%%%%%%%%%%%%%
\subsection{Restrictions on entanglement groups}
%%%%%%%%%%%%%%%%%%%%%%%%%%%%%%%%%%%%%%%%%%%%%
The isomorphism and no-sharing theorems of sections \ref{sect:isomorphism} and \ref{sect:noshare} restrict the structure of the
entanglement groups and can lead to further consequences.  Here we give examples of a few of these consequences.

Consider a tripartite system $A,B,C$ and suppose $A$ and $B$ are maximally entangled.  We take this to mean that $E_{AB}$ is as
large as it can be, $E_{AB} = PSU(d)$ for the case of equal Hilbert space dimensions $d_A = d_B = d$.  What does this imply for the entanglement groups?
Since the center of $PSU(d)$ is trivial, and different two-party entanglement groups must commute with each
other, we immediately see that $E_{AC}$ and $E_{BC}$ must be trivial.  We can also show that there is no three-party entanglement.
To see this suppose we have an element of the full entanglement group $f = (f_A,f_B,f_C) \in F_{ABC}$.  Since $E_{AB} = PSU(d)$, for every $f_{A} \in \pi_{A}(f)$,
there is an element $g=(f_A,\phi(f_A),1) \in F_{ABC}$ where $\phi$ is an isomorphism.  But then we see that  $ g^{-1}\cdot f=  (1,\phi(f_A)^{-1} f_B,f_C) \in E_{BC}$, so $f$ can be written  as a product of an element
of $E_{AB}$ with an element of $E_{BC}$.
\be
f = (f_A,\phi(f_A),1) \cdot (1,\phi(f_A)^{-1} f_B,f_C)
\ee
Since any $f \in F_{ABC}$ can be written in this way there is no three-party entanglement.  Overall we've shown that for equal Hilbert space
dimensions maximal entanglement between $A$ and $B$ implies
\be
\hbox{\rm $F_{ABC} = E_{AB} = PSU(d)$ with $E_{AC},E_{BC},E_{ABC}$ all trivial}
\ee
This is the usual statement that maximal entanglement is monogamous.  It's equivalent to a statement about stabilizers,
\be
\hbox{\rm $S_{ABC} = S_{AB} \times S_C$ with $S_{AB} = PSU(d)$ and $S_A,S_B$ trivial}.
\ee
If the Hilbert space dimensions are not equal, say $d_A < d_B$, then maximal entanglement of system $A$ with system $B$ still means $E_{AB} = PSU(d_A)$.
In this case it is possible
to have a non-trivial $E_{BC}$ as long as it commutes with $PSU(d_A)$.  However $E_{AC}$ must be trivial.  Also the proof that $E_{ABC}$ is trivial still goes through so there is
no three-party entanglement.

As a concrete example of maximal entanglement consider a system with a qubit $A$, a qu-4it $B$, and a qubit $C$. Consider the state
\begin{equation}
\ket{\psi} = a \big(\vert 0 \rangle_{A} \vert 0 \rangle_{B} + \vert 1 \rangle_{A} \vert 1 \rangle_{B}\big) \vert 0 \rangle_{C} + b \big(\vert 0 \rangle_{A} \vert 2 \rangle_{B} + \vert 1 \rangle_{A} \vert 3 \rangle_{B}\big) \vert 1 \rangle_{C}
\end{equation}
This state has a two-party stabilizer $S_{AB} = U(1) \times U(2)$ which acts by
\be
e^{i \theta} \begin{pmatrix}
\alpha & \beta \\
-\bar{\beta} &\bar{\alpha}
\end{pmatrix}
\otimes
\begin{pmatrix}
\bar{\alpha}& \bar{\beta} & & \\
-\beta & \alpha & & \\
& & \bar{\alpha}& \bar{\beta} \\
& & -\beta & \alpha
\end{pmatrix}
\otimes
\begin{pmatrix}
1 & 0 \\
0 & 1
\end{pmatrix} \qquad \vert \alpha \vert^2 + \vert \beta \vert^2 = 1
\ee
This gives $E_{AB} = PSU(2)$.  For generic $a,b$ the state has $E_{BC} = U(1)$ but for $a=b$ it is enhanced to $E_{BC} = PSU(2)$. 
Inside the $SU(4)$ of the qu-4it the groups $\pi_{B}(E_{AB})$ and $\pi_{B}(E_{BC})$ are embedded so that they commute. In addition, in this case, 
$\pi_{B}(E_{AB}) \cap \pi_{B}(E_{BC}) = \lbrace \, 1 \, \rbrace$.  We also find that $E_{AC}$ and $E_{ABC}$ are trivial.

Here is another example of how general considerations restrict the structure of the entanglement groups.  Suppose we're working with the
generalized GHZ state (\ref{sghz}), but without knowing the exact state vector. Instead suppose we only know that we have a state of three qubits with one-party
stabilizers $S_{A},S_{B},S_{C}$ that are just phases times the identity operator and with two-party stabilizers given by (\ref{ghz2p1}), (\ref{ghz2p2}), (\ref{ghz2p3}).
One then asks: is it possible for such a state to also have some non-trivial three-party entanglement?  If there is an additional element of $S_{ABC}$ it must be of
the form $u_{A}\otimes u_{B} \otimes u_{C}$ and in the normalizer of $S_{AB},S_{AC},S_{BC}$. A simple computation shows that an element of the form 
\begin{equation}
\begin{pmatrix}
0 &e^{i\delta_1}  \\
-e^{-i\delta_1} & 0
\end{pmatrix}
\otimes
\begin{pmatrix}
0 &e^{i\delta_2}  \\
-e^{-i\delta_2} & 0
\end{pmatrix}
\otimes
\begin{pmatrix}
0 &e^{i\delta_3}  \\
-e^{-i\delta_3} & 0
\end{pmatrix}
\end{equation}
is in the normalizer of the two-party entanglement: under conjugation the matrices (\ref{ghz2p1}), (\ref{ghz2p2}), (\ref{ghz2p3}) are invariant
up to a phase.  However this element under multiplication by the known $S_{AB},S_{AC},S_{BC}$ is equivalent
to an element of the same form but with $\delta_1=\delta_2 =\delta_3 =\delta$. Can we have two such elements in $S_{ABC}$? If we have two such
elements with parameters $\delta$ and $\tilde{\delta}$ then multiplying one by the other and using some elements of $S_{AB},S_{AC},S_{BC}$
can generate an element of $S_{A},S_{B},S_{C}$ that is non-trivial (not proportional to the identity). That contradicts one of the assumptions,
so we conclude that we can have at most one additional stabilizer of the form\footnote{We fixed the phase so that the stabilizer squares to the identity.}
\begin{equation}
i \begin{pmatrix}
0 &e^{i\delta}  \\
-e^{-i\delta} & 0
\end{pmatrix}
\otimes
\begin{pmatrix}
0 &e^{i\delta}  \\
-e^{-i\delta} & 0
\end{pmatrix}
\otimes
\begin{pmatrix}
0 &e^{i\delta}  \\
-e^{-i\delta} & 0
\end{pmatrix}
\end{equation}
Such a stabilizer would indeed give a three-party entanglement group $E_{ABC} = {\mathbb Z}_2$.  The state which has this additional stabilizer
is of the form (\ref{sghz}) with $a = \pm i b e^{3i\delta}$. 

%%%%%%%%%%%%%%%%%%%%%%%%%%%%%
\section{Mixed states\label{sect:MixedStates}}
%%%%%%%%%%%%%%%%%%%%%%%%%%%%%
So far we have introduced entanglement groups for pure states.  In this section we consider the extension to mixed states.

To illustrate the ideas, we start with a two-party density matrix $\rho_{AB}$.  We purify $\rho_{AB}$ to a state $\vert \psi \rangle_{ABC}$ and introduce an appropriate notion of an entanglement
group $\EtildeAB$ associated with the purification.  After doing this, and showing that the entanglement groups are well-defined (independent of the purification), we return to relate the entanglement groups
directly to properties of $\rho_{AB}$.  The generalization to multipartite density matrices is straightforward and will be mentioned in passing.

Consider a bipartite Hilbert space ${\cal H}_A \otimes {\cal H}_B$.  An ensemble of states
\be
\lbrace \vert \ell \rangle_{AB} \in {\cal H}_A \otimes {\cal H}_B \rbrace \qquad \ell = 1,\ldots,L
\ee
with probabilities $p_\ell$ corresponds to a density matrix
\be
\rho_{AB} = \sum_{\ell = 1}^L p_\ell \, \vert \ell \rangle_{AB} \, {}_{AB} \langle \ell \vert
\ee
We're assuming the states $\vert \ell \rangle_{AB}$ are normalized, but they need not be orthogonal or complete in ${\cal H}_A \otimes {\cal H}_B$.

The density matrix can be purified by introducing an auxiliary Hilbert space ${\cal H}_C$ with an orthonormal basis $\lbrace \vert \ell \rangle_C \rbrace$ and
considering the state
\be
\vert \psi \rangle = \sum_{\ell = 1}^L \sqrt{p_\ell} \, \vert \ell \rangle_{AB} \otimes \vert \ell \rangle_C
\ee
Since ${}_C \langle \ell \vert \ell' \rangle_C = \delta_{\ell \ell'}$ we have ${\rm Tr}_C \left( \vert \psi \rangle \langle \psi \vert \right) = \rho_{AB}$.

Given the purification $\vert \psi \rangle_{ABC}$, we could define a variety of entanglement groups, as quotients of the various stabilizers $S_{ABC}$, $S_{AB}$, $S_{AC}$, $S_{BC}$,
$S_A$, $S_B$, $S_C$.  However we must ask which quotients have operational significance in terms of the original density matrix.  In this regard, remember that we do not
observe system $C$ and will eventually trace over it.  Therefore, the appropriate notion of a stabilizer is an element $\widetilde{s}_{AB} \in U(d_A) \times U(d_B)$ with the
property that it leaves the purification invariant up to a unitary on system $C$.
\be
\big(\widetilde{s}_{AB} \otimes \identity_C\big) \vert \psi \rangle = \big(\identity_{AB} \otimes u_C\big) \vert \psi \rangle
\ee
When tracing over system $C$, note that $u_C$ drops out, so such a stabilizer will leave the density matrix invariant.  In this way a stabilizer for the density matrix $\widetilde{s}_{AB}$ corresponds
to a stabilizer for the purification $\widetilde{s}_{AB} \otimes u_C^{-1}$.

We can obtain the stabilizer groups for $\rho_{AB}$, denoted with a tilde, by acting with $AB$ projection operators on the stabilizer groups for the purification.\footnote{The projection is defined by
$\pi_{AB}(u_A,u_B,u_C) = (u_A,u_B)$.}
\bea
\nonumber
&& \widetilde{S}_{AB} = \pi_{AB}\big(S_{ABC}\big) \\
&& \widetilde{S}_A = \pi_{AB}\big(S_{AC}\big) \\
\nonumber
&& \widetilde{S}_B = \pi_{AB}\big(S_{BC}\big)
\eea
Given these stabilizer groups, we can construct an entanglement group for the density matrix
\be
\label{EtildeAB}
\EtildeAB = \widetilde{S}_{AB} / \big(\widetilde{S}_A \times \widetilde{S}_B\big) = \pi_{AB}\big(S_{ABC} / \left(S_{AC} \cdot S_{BC}\right)\big)
\ee
This is the entanglement group that we associate with a two-party density matrix $\rho_{AB}$.  To make a few comments:
\begin{itemize}
\item
From the point of view of the purification, $\widetilde{S}_{AB}$ allows a unitary transformation on system $C$.  This means it includes both
the  two-party $AB$ entanglement and the three-party $ABC$ entanglement that is present in the purification.  If we had access to system $C$ we could distinguish these possibilities
by examining the entanglement groups associated with the pure state $\vert \psi \rangle$, namely
\bea
\nonumber
&& E_{AB} = S_{AB} / \left(S_A \times S_B\right) \\[5pt]
\label{distinguish}
&& E_{ABC} = S_{ABC} / \left(S_{AB} \cdot S_{AC} \cdot S_{BC}\right)
\eea
However, the assumption is that we do not have access to system $C$, so if we are just given the density matrix, we cannot make this distinction in a meaningful way.  Instead, the only
physically meaningful entanglement group associated with $\rho_{AB}$ is $\EtildeAB$.
\item
As mentioned above, the groups $E_{AB}$ and $E_{ABC}$ do not capture all of the
entanglement involving $A$ and $B$ which may be present in the purification.  There could also be stabilizers $S_{A(BC)}$ that act with a unitary on the combined $(BC)$ system.  We will return to this point in section \ref{sect:A(BC)}.
\item
The generalization to multipartite mixed states is straightforward.  For example, consider a tripartite density matrix $\rho_{ABC}$.  This has a purification $\vert \psi \rangle_{ABCD}$
and we can define entanglement groups
\bea
\nonumber
&&\EtildeAB = \pi_{ABC}\big(S_{ABD} / \left(S_{AD} \cdot S_{BD}\right)\big) \\
\nonumber
&&\EtildeAC = \pi_{ABC}\big(S_{ACD} / \left(S_{AD} \cdot S_{CD}\right)\big) \\
\label{EtildeABC}
&&\EtildeBC = \pi_{ABC}\big(S_{BCD} / \left(S_{BD} \cdot S_{CD}\right)\big) \\
\nonumber
&&\EtildeABC = \pi_{ABC}\big(S_{ABCD} / \left(S_{ABD} \cdot S_{ACD} \cdot S_{BCD}\right)\big)
\eea
\end{itemize}

Since the definitions (\ref{EtildeAB}), (\ref{EtildeABC}) are in terms of a purification, we must show that these entanglement groups are
well-defined, that is, that they are independent of the choice of purification.  Every density
matrix has a minimal purification, constructed as follows.  Start by diagonalizing the density matrix.
\be
\rho_{AB} = \sum_{i = 1}^r p_i \, \vert i \rangle_{AB} {}_{AB} \langle i \vert
\ee
Here $r$ is the rank of $\rho_{AB}$.  A minimal purification is then
\be
\vert \psi \rangle = \sum_{i = 1}^r \sqrt{p_i} \, \vert i \rangle_{AB} \otimes \vert i \rangle_C
\ee
Any other purification of $\rho_{AB}$ is related to this one by \cite{nielsen2010quantum}
\begin{enumerate}
\item
Embedding the orthonormal vectors $\vert i \rangle_C$ into a (sufficiently large) Hilbert space ${\cal H}_{C'}$;
\item
Performing a unitary transformation on ${\cal H}_{C'}$.
\end{enumerate}
For later reference, two different ensembles $\lbrace (p_\ell,\, \vert \ell \rangle_{AB}) \rbrace$ and $\lbrace (q_m,\, \vert m \rangle_{AB}) \rbrace$ can generate
exactly the same density matrix \cite{nielsen2010quantum}.  The ensembles lead to nominally different purifications,
\bea
\nonumber
&& \vert \psi \rangle = \sum_{\ell = 1}^L \sqrt{p_\ell} \, \vert \ell \rangle_{AB} \otimes \vert \ell \rangle_C \\
\label{freedom}
&& \vert \psi' \rangle = \sum_{m = 1}^M \sqrt{q_m} \, \vert m \rangle_{AB} \otimes \vert m \rangle_{C'}
\eea
which are however related in the manner described (if necessary, by embedding the vectors in the larger of ${\cal H}_C$ and ${\cal H}_{C'}$, then performing
a unitary transformation on the auxiliary Hilbert space).

Are the entanglement groups associated with a density matrix well-defined?  For $\rho_{AB}$, entanglement groups are defined in terms of stabilizer groups such as
$S_{AC}$ and $S_{ABC}$.  These groups transform in a simple way under unitary transformations on system $C$.\footnote{Likewise $S_{A(BC)}$ behaves
nicely under unitary transformations on the combined $(BC)$ system, which includes unitary transformations on system $C$ as a subgroup.}
\bea
\nonumber
&& \vert \psi \rangle \rightarrow \left(\identity_A \otimes \identity_B \otimes u_C\right) \vert \psi \rangle \equiv {\cal U}_C \vert \psi \rangle \\
&& S_{AC} \rightarrow {\cal U}_C S_{AC} \, {\cal U}_C^\dagger \\
\nonumber
&& S_{ABC} \rightarrow {\cal U}_C S_{ABC} \, {\cal U}_C^\dagger
\eea
These transformations leave the abstract groups $S_{AC}$ and $S_{ABC}$ invariant, although they change the way they're embedded in
$U(d_A) \times U(d_B) \times U(d_C)$.  So the entanglement groups are well-defined, independent of the choice of purification.

At this stage we have a definition of entanglement groups for a density matrix $\rho_{AB}$ in terms of stabilizers of the purification $\vert \psi \rangle$.  What do these stabilizers
mean in terms of the density matrix itself?  To illustrate, consider a stabilizer
\be
s_{ABC} = u_A \otimes u_B \otimes u_C \in S_{ABC}
\ee
which by definition means
\be
s_{ABC} \vert \psi \rangle = e^{i \theta} \vert \psi \rangle
\ee
or equivalently
\be
\big(u_A \otimes u_B \otimes \identity_C\big) \vert \psi \rangle = e^{i \theta} \big(\identity_{AB} \otimes u_C^{-1}\big) \vert \psi \rangle
\ee
Tracing out system $C$, and denoting $\widetilde{s}_{AB} = u_A \otimes u_B$, we have
\be
\label{RhoStabilizer}
\widetilde{s}_{AB} \, \rho_{AB} \, \widetilde{s}_{AB}^{\,\,\dagger} = \rho_{AB}
\ee
So the density matrix is invariant under conjugation by $\widetilde{s}_{AB}$.

Conversely, given a stabilizer $\widetilde{s}_{AB}$ satisfying (\ref{RhoStabilizer}), it can be lifted to an element of $S_{ABC}$ for the purification.  To see this, suppose
$\lbrace (p_\ell,\, \vert \ell \rangle_{AB}) \rbrace$ is an ensemble of states that correspond to the density matrix $\rho_{AB}$.  Given a stabilizer (\ref{RhoStabilizer}), it
follows that $\lbrace (p_\ell,\, \widetilde{s}_{AB} \vert \ell \rangle_{AB}) \rbrace$ corresponds to exactly the same density matrix.  Then, due to the uniqueness discussed
around (\ref{freedom}), the two purifications are related by a unitary transformation $u_C$ on the auxiliary system.  That is, there is a stabilizer for the purification
$\widetilde{s}_{AB} \otimes u_C \in S_{ABC}$.

This means we can define the group $\widetilde{S}_{AB}$ as the subgroup of $U(d_A) \times U(d_B)$ with the property that the density matrix is invariant under conjugation
by every element of $\widetilde{S}_{AB}$.
\be
\label{def1}
\widetilde{s}_{AB} \, \rho_{AB} \, \widetilde{s}_{AB}^{\,\,\dagger} = \rho_{AB} \qquad \hbox{\rm for all $\widetilde{s}_{AB} \in \widetilde{S}_{AB}$}
\ee
This has a normal subgroups $\widetilde{S}_A$, $\widetilde{S}_B$ which act trivially on one of the subsystems.
\bea
\nonumber
&& \widetilde{S}_A = \left\lbrace \hbox{$u_A \otimes \identity_B \in \widetilde{S}_{AB}$ such that $\big(u_A \otimes \identity_B\big) \rho_{AB} \big(u_A \otimes \identity_B\big)^\dagger = \rho_{AB}$}\right\rbrace \\
\label{def2}
&& \widetilde{S}_B = \left\lbrace \hbox{$\identity_A \otimes u_B \in \widetilde{S}_{AB}$ such that $\big(\identity_A \otimes u_B\big) \rho_{AB} \big(\identity_A \otimes u_B\big)^\dagger = \rho_{AB}$}\right\rbrace
\eea
This lets us construct the entanglement group
\be
\label{def3}
\EtildeAB = \widetilde{S}_{AB} / (\widetilde{S}_A \times \widetilde{S}_B)
\ee
directly from the density matrix, without reference to a purification.  Note that this construction is operationally meaningful, in the sense that it involves a physical operation (unitary conjugation) on the density matrix.

%%%%%%%%%%%%%%%%%%%%%%%%%%%%%
\subsection{Structure of mixed-state entanglement groups\label{sect:properties}}
%%%%%%%%%%%%%%%%%%%%%%%%%%%%%
Having defined entanglement groups for density matrices, we'd like to explore some of their properties.  In what follows we're mostly concerned
with the action of the entanglement groups on individual systems.  For this reason we introduce projection operators $\pi_A, \, \pi_B, \, \pi_C, \, \ldots$ that project
onto systems $A, \, B, \, C, \, \ldots$.  For example, given $g = (g_A,g_B) \in \EtildeAB$, we have $\pi_A(g) = g_A$.

There are two main questions we'd like to address.
\begin{enumerate}
\item
Consider a two-party density matrix $\rho_{AB}$.  Let $g_A \in \pi_A(\EtildeAB)$.  Is there a unique element $g_{AB} \in \EtildeAB$ that projects
to give $g_A$?  If so, how are the actions of $g_{AB}$ on systems $A$ and $B$ related?
\item
Consider a three-party density matrix $\rho_{ABC}$.  Let $g_A \in \pi_A(\EtildeAB)$.  Could $g_A$ appear in the projection of any other entanglement groups?
That is, could $g_A$ also be an element of $\pi_A(\EtildeAC)$ or of $\pi_A(\EtildeABC)$?
\end{enumerate}
In what follows we show that:
\begin{enumerate}
\item
There is a unique $g_{AB} \in \EtildeAB$ that projects to any given $g_A \in \pi_A(\EtildeAB)$.  It acts isomorphically on systems $A$ and $B$.
\item
A transformation $g_A$ can appear in both $\pi_A(\EtildeAB)$ and $\pi_A(\EtildeAC)$, but only if it is an element of the center of both groups.
It cannot appear in the projection of both a two-party entanglement group (meaning $\EtildeAB$, $\EtildeAC$) and a three-party entanglement group $\EtildeABC$.
\end{enumerate}
The discussion here closely follows the treatment of pure states in section \ref{sect:monogamy}.

%%%%%%%%%%%%%%%%%%%%%%%%%%%%%
\subsubsection{Isomorphism theorems\label{sect:MixedStateIsomorphism}}
%%%%%%%%%%%%%%%%%%%%%%%%%%%%%
For a two-party density matrix $\rho_{AB}$, we wish to show that there is a unique $g_{AB} \in \EtildeAB$ that projects to any given $g_A \in \pi_A(\EtildeAB)$, and
that $g_{AB}$ acts isomorphically on systems $A$ and $B$.  We begin with some suggestive observations at the level of stabilizers, then turn to entanglement groups.

Purify $\rho_{AB}$ to a state $\vert \psi \rangle_{ABC}$, and consider two stabilizers which happen to have the same action on system $A$.
\bea
\nonumber
&& s_1 = (s_A,s_B,s_C) \in S_{ABC} \\[5pt]
&& s_2 = (s_A,s_B',s_C') \in S_{ABC}
\eea
It follows that
\be
s_1 s_2^{-1} = (\identity_A, s_B s_B'^{\,-1}, s_C s_C'^{\,-1})
\ee
is a stabilizer which only acts on systems $B$ and $C$, so $s_1 s_2^{-1} \in S_{BC}$.  When we quotient to form $\EtildeAB$, elements of $S_{BC}$ are mapped to the identity.
So $s_1$ and $s_2$ represent the same equivalence class in $\EtildeAB$.  This means $s_A$ determines a unique element $g_{AB} \in \EtildeAB$, and hence a unique
element $g_B \in \pi_B(\EtildeAB)$.

This observation about stabilizers is suggestive, but does not quite establish the desired result.  We'd like to make a statement, not about $s_A$, but rather about an
equivalence class $[s_A] \in \pi_A(\EtildeAB)$.  For this we first introduce some convenient notation.  Let
\be
G = S_{ABC} \qquad\quad N = S_{AC} \cdot S_{BC}
\ee
with projections
\bea
\nonumber
G_A = \pi_A(G) \qquad\quad N_A = \pi_A(N) \\[5pt]
G_B = \pi_B(G) \qquad\quad N_B = \pi_B(N)
\eea
The entanglement group $\EtildeAB$ is constructed as
\be
\EtildeAB = \pi_{AB}(G / N)
\ee
We claim there are group isomorphisms
\bea
\nonumber
&& \theta_A \, : \, \EtildeAB \rightarrow G_A / N_A \\
&& \theta_B \, : \, \EtildeAB \rightarrow G_B / N_B
\eea
which make
\be
\label{isomorphism}
\EtildeAB \approx G_A / N_A \approx G_B / N_B
\ee
This means we can present $\EtildeAB$ in the form
\be
\EtildeAB = \left\lbrace \, \Big(\,\theta_A(t),\, \theta_B(t)\,\Big) \, \Big\vert \, t \in \EtildeAB \, \right\rbrace
\ee
So there is a unique element of $\EtildeAB$ that projects to any given $[s_A]$ or $[s_B]$, and it acts isomorphically on systems
$A$ and $B$.

The proof of (\ref{isomorphism}) is as follows.  The isomorphism $G_A / N_A \approx G_B / N_B$ is given as Theorem 1
in appendix \ref{appendix:Goursat}.  We also have an isomorphism\footnote{Proof: let $H = \left\lbrace \, \Big( \, [s_A],\,[s_B],\,[s_C] \, \Big) \, \in  {G_A \over N_A} \times {G_B \over N_B} \times {G_C \over N_C} \, \Big\vert \,
\big(s_A,s_B,s_C\big) \in G \, \right\rbrace$ and define a map $\phi \, : \, G \rightarrow H$ by $\phi\big((s_A,s_B,s_C)\big) = \big( \, [s_A],\,[s_B],\,[s_C] \, \big)$.  Then ${\rm ker} \, \phi = N$ and
an isomorphism theorem gives $G / {\rm ker} \phi \approx H$.}
\be
G / N \approx \left\lbrace \, \Big( \, [s_A],\,[s_B],\,[s_C] \, \Big) \, \in  {G_A \over N_A} \times {G_B \over N_B} \times {G_C \over N_C} \, \Big\vert \,
(s_A,s_B,s_C) \in G \, \right\rbrace
\ee
Acting on this with $\pi_{AB}$ we have the isomorphisms given in (\ref{isomorphism}).

%%%%%%%%%%%%%%%%%%%%%%%%%%%%%
\subsubsection{Restrictions on sharing\label{sect:MixedStateSharing}}
%%%%%%%%%%%%%%%%%%%%%%%%%%%%%
Now we turn to a three-party density matrix $\rho_{ABC}$ and ask: can a given element $g_A \in \pi_A(\EtildeAB)$ appear in the
projection of more than one entanglement group?  Could it also appear in $\pi_A(\EtildeABC)$?  What about $\pi_A(\EtildeAC)$?

We begin by disposing of the first possibility: if $g_A$ appears in the projection of a two-party entanglement group, it cannot also
appear in the projection of a three-party group.  To see this, suppose there are stabilizers
\bea
\nonumber
&& s_1 = \big(s_A,s_B,\identity_C,s_D\big) \in S_{ABD} \\
&& s_2 = \big(s_A,s_B',s_C',s_D'\big) \in S_{ABCD}
\eea
Then $s_1^{-1} s_2 = \big(\identity_A, s_B^{-1}s_B', s_C', s_D^{-1} s_D' \big) \in S_{BCD}$, which means we can write $s_2$ as a product of an element of $S_{ABD}$ with an element of $S_{BCD}$.
\be
s_2 = \big(s_A,s_B,\identity_C,s_D\big) \cdot \big(\identity_A, s_B^{-1}s_B', s_C', s_D^{-1} s_D' \big)
\ee
Recall that $\EtildeABC$ involves a quotient by $S_{ABD}$ and $S_{BCD}$.  So written in this way, we see that $s_2$ maps to the identity element of $\EtildeABC$.

Could $g_A$ also appear in $\pi_A(\EtildeAC)$?  This is a bit more subtle.  Consider two stabilizers
\bea
\nonumber
&& s_1 = \big(s_A,s_B,\identity_C,s_D\big) \in S_{ABD} \\
&& s_2 = \big(s_A',\identity_B,s_C',s_D'\big) \in S_{ACD}
\eea
Their group commutator is
\be
s_1 s_2 s_1^{-1} s_2^{-1} = \big(s_A s_A' s_A^{-1} s_A'^{\,-1}, \identity_B, \identity_C, s_D s_D' s_D^{-1} s_D'^{\,-1}\big) \in S_{AD}
\ee
The commutator becomes trivial in both $\EtildeAB$ and $\EtildeAC$, since those groups involve a quotient by $S_{AD}$.
Projecting onto system $A$, we see that $\pi_A(\EtildeAB)$ and $\pi_A(\EtildeAC)$ commute element-by-element.  This means that a single
element $g_A$ could appear in both groups, but only if it is an element of the center of both.
\be
\label{Z}
g_A \in Z\big(\pi_A(\EtildeAB)\big) \cap Z\big(\pi_A(\EtildeAC)\big)
\ee
In this sense, $g$-entanglement is not strictly monogamous: an element of the common center (\ref{Z}) can be thought of as reflecting
either $AB$ entanglement or $AC$ entanglement.

%%%%%%%%%%%%%%%%%%%%%%%%%%%%%
\section{Separable states\label{sect:separable}}
%%%%%%%%%%%%%%%%%%%%%%%%%%%%%
When a density matrix $\rho_{AB}$ is separable, there are only classical correlations between measurements on systems $A$ and
$B$.  This is usually taken to mean there's no entanglement between $A$ and $B$ \cite{Werner:1989zz}.  Our goals here are to translate
separability into the language of stabilizer groups, and to understand what it means from this perspective.
We will see that, although separability places severe restrictions, it does not quite imply that $\EtildeAB$ is trivial,
because the purification could have three-party entanglement with the auxiliary system.

Consider a bipartite Hilbert space ${\cal H}_A \otimes {\cal H}_B$.  A separable density matrix $\rho_{AB}$ has the form
\be
\label{separable}
\rho_{AB} = \sum_\ell p_\ell \, \rho_A^{(\ell)} \otimes \rho_B^{(\ell)}
\ee
with $p_\ell > 0$ and $\sum_\ell p_\ell = 1$.  We can take $\rho_A^{(\ell)}$ and $\rho_B^{(\ell)}$ to be pure-state density matrices \cite{Horodecki:1997vt,Horodecki:2009zz}, so that
\be
\label{separable2}
\rho_{AB} = \sum_{\ell = 1}^L p_\ell \, \big(\vert \ell \rangle_A \, {}_A \langle \ell \vert \big) \otimes \big(\vert \ell \rangle_B \, {}_B \langle \ell \vert \big)
\ee
It follows from Carath\'odory's theorem that $L \leq d_{AB}^2 = (d_A d_B)^2$ \cite{Horodecki:1997vt,Horodecki:2009zz}, where the dimensions $d_A = {\rm dim} \, {\cal H}_A$,
$d_B = {\rm dim} \, {\cal H}_B$ are assumed finite.

In what follows we first discuss purification of $\rho_{AB}$, by introducing an auxiliary Hilbert space ${\cal H}_C$.  Then we examine
the stabilizers and entanglement groups associated with the purification.  We will consider both
\begin{itemize}
\item
the group $E_{AB}$, which captures entanglement between $A$ and $B$ while leaving out the auxiliary system;
\item
the group $E_{A(BC)}$, which captures entanglement between $A$ and the combined $(BC)$ system.
\end{itemize}
Along the way we will see that, although the entanglement group $\EtildeAB$ associated with the density matrix itself can be non-trivial,
it can only arise from three-party entanglement between $A$, $B$ and the purifying system $C$.

%%%%%%%%%%%%%%%%%%%%%%%%%%%%%
\subsection{Purifying a separable density matrix}
%%%%%%%%%%%%%%%%%%%%%%%%%%%%%
Given a separable density matrix $\rho_{AB}$, we can purify it to a pure state of a tripartite system $ABC$ by introducing
\be
\label{pure}
\vert \psi \rangle = \sum_\ell \sqrt{p_\ell} \, \vert \ell \rangle_A \otimes \vert \ell \rangle_B \otimes \vert \ell \rangle_C
\ee
Here $\lbrace \vert \ell \rangle_C \rbrace$ is a basis for the Hilbert space ${\cal H}_C$, which we take to be orthonormal.
\be
{}_C\langle \ell \vert \ell' \rangle_C = \delta_{\ell \ell'} \qquad \ell,\,\ell' = 1,\ldots,L
\ee
A purification of this form is only possible for a separable density matrix.  A general superposition of factorized states
\be
\vert \psi \rangle = \sum_\ell \alpha_\ell \, \vert \ell \rangle_A \otimes \vert \ell \rangle_B \otimes \vert \ell \rangle_C \qquad \alpha_\ell \in {\amsbb C}
\ee
leads to
\be
\rho_{AB} = \sum_\ell \vert \alpha_\ell \vert^2 \big(\vert \ell \rangle_A \, {}_A \langle \ell \vert \big) \otimes \big(\vert \ell \rangle_B \, {}_B \langle \ell \vert \big)
\ee
So separability --  including the fact that the coefficients $p_\ell$ in (\ref{separable}) are all positive -- is essential.

A few comments on this purification are in order.  First, note that there's no reason to expect the vectors $\vert \ell \rangle_A$ to be linearly independent in ${\cal H}_A$, nor is there any reason to expect them to
span all of ${\cal H}_A$.  Likewise for $\vert \ell \rangle_B$ and ${\cal H}_B$.  Next, note that the auxiliary Hilbert space ${\cal H}_C$ we introduced has dimension ${\rm dim} \, {\cal H}_C = L$.  The minimum dimension required to purify
a density matrix is equal to the rank of the density matrix.  Generically $L$ is bigger than the rank of $\rho_{AB}$, so generically we're using a non-minimal purification.
However, as discussed previously, it is related to any other purification by a unitary transformation on system $C$.

%%%%%%%%%%%%%%%%%%%%%%%%%%%%%
\subsection{Entanglement group for a separable density matrix\label{subsec:EntGpForSeparable}}
%%%%%%%%%%%%%%%%%%%%%%%%%%%%%
Suppose a density matrix is separable.  What does separability imply for the associated entanglement group?  One might expect that
\be
\label{EtildeAB2}
\EtildeAB = \pi_{AB}\big(S_{ABC} / \left(S_{AC} \cdot S_{BC}\right)\big)
\ee
should be trivial.  It turns out this is not quite the case.

To understand what separability implies, we first purify $\rho_{AB}$ to a state of the form (\ref{pure}).  In appendix \ref{appendix:twoparty} we show that any two-party stabilizer of the purification $s_{AB}$ can be written as a product
$s_{AC} s_{BC}$.  Such stabilizers can give the pure state $\vert \psi \rangle$ a two-party entanglement group, since there is no reason for
\be
E_{AB} = S_{AB} / \left(S_A \times S_B\right)
\ee
to be trivial.  However such stabilizers do not contribute to the entanglement group for the density matrix $\EtildeAB$: due to the quotient
by $S_{AC} \cdot S_{BC}$, such stabilizers map to the identity.  Although two-party
entanglement in the purification does not contribute to $\EtildeAB$, it is still possible for $\EtildeAB$ to be non-trivial, because there could be
three-party $g$-entanglement in the purification.

To illustrate this, consider the generalized GHZ state
\be
\vert \psi \rangle = a \vert 000 \rangle + b \vert 111 \rangle
\ee
For generic values of $a$ and $b$ the entanglement groups are, from section \ref{sect:examples} \footnote{Here $\varnothing$ denotes the trivial group.  Recall that the generalized GHZ state has tri-partite $u$-entanglement,
which in the $g$-entanglement approach corresponds to having non-trivial groups $E_{A(BC)} = E_{B(AC)} = E_{C(AB)} = U(1)$.  However it does not have bi-partite $u$-entanglement
since the reduced density matrices are separable.  For a discussion of $g$-entanglement vs.\ $u$-entanglement in GHZ see section \ref{sect:usualnotion}.}
\be
E_{AB} = U(1) \qquad\quad E_{ABC} = \varnothing
\ee
Tracing out system $C$ leads to a separable density matrix
\be
\rho_{AB} = \vert a \vert^2 \, \vert 00 \rangle \langle 00 \vert + \vert b \vert^2 \, \vert 11 \rangle \langle 11 \vert
\ee
with $\EtildeAB = \varnothing$.  The absence of three-party entanglement in the purification means that, as argued in general above,
$\EtildeAB$ is trivial.

By contrast, the standard GHZ state with $a = b = 1/\sqrt{2}$ has an additional three-party stabilizer.
\be
x = \begin{pmatrix}
0 & 1 \\
1 & 0
\end{pmatrix}
\otimes
\begin{pmatrix}
0 & 1 \\
1 & 0
\end{pmatrix}
\otimes
\begin{pmatrix}
0 & 1 \\
1 & 0
\end{pmatrix} \in S_{ABC}
\ee
This enhances $E_{ABC}$ from $\varnothing$ to ${\amsbb Z}_2$.  As a result, $\EtildeAB$ is also enhanced to ${\amsbb Z}_2$.  This illustrates the fact that three-party entanglement
with the purifying system can lead to a non-trivial $\EtildeAB$.

As a further comment, one could imagine starting with a pure state $\vert \psi \rangle_{ABC}$ and obtaining $\rho_{AB}$ as the reduced density matrix that describes measurements
performed on the $AB$ system.  In this case system $C$ is in principle observable, so as discussed around (\ref{distinguish}), more information is in principle available.  In particular, in addition to
the entanglement group $\EtildeAB$ that is associated with the density matrix, we could consider the two-party entanglement group $E_{AB}$ that is associated with the pure state
$\vert \psi \rangle_{ABC}$.
(Recall that $E_{AB}$ is the entanglement group built from stabilizers that do not act on system $C$.  Given the density matrix $\rho_{AB}$
this is not meaningful, but given the pure state $\vert \psi \rangle_{ABC}$ it is.)  $E_{AB}$ has a very particular structure.
We have seen that, if $\rho_{AB}$ is separable, all $AB$ stabilizers
are built up as a product of $AC$ and $BC$ stabilizers.  As shown in appendix \ref{appendix:twoparty}, this means that in the $g$-entanglement scheme $AB$ entanglement can be thought of as a combination of $AC$ entanglement with $BC$ entanglement.

%%%%%%%%%%%%%%%%%%%%%%%%%%%%%
\subsection{Structure of $A$ -- $(BC)$ entanglement in the purification\label{sect:A(BC)}}
%%%%%%%%%%%%%%%%%%%%%%%%%%%%%
Here we consider entanglement between system $A$ and the combined $(BC)$ system, captured by the group $E_{A(BC)}$.  What does separability have to say about such entanglement?

First, note that $\rho_{AB}$ is separable if and only if there is a controlled unitary\footnote{a unitary transformation ${\cal U}_{BC}$ that is block diagonal, as in the second
factor of (\ref{block}), controlled by the state of system $C$} that acts on the combined $(BC)$ system and
completely disentangles $B$.  To construct ${\cal U}_{BC}$, consider the purification of a separable density matrix to
\be
\label{pure2}
\vert \psi \rangle = \sum_{\ell = 1}^L \sqrt{p_\ell} \, \vert \ell \rangle_A \otimes \vert \ell \rangle_B \otimes \vert \ell \rangle_C
\ee
Choose any fixed vector $\vert \chi \rangle_B \in {\cal H}_B$ and choose a set of unitary matrices $u_B^{(\ell)}$ with the property that for each $\ell$
\be
u_B^{(\ell)} \vert \ell \rangle_B = \vert \chi \rangle_B
\ee
Since we don't specify how $u_B^{(\ell)}$ acts on any vector orthogonal to $\vert \ell \rangle_B$, there is quite a bit of freedom in choosing these matrices.\footnote{There is
a unitary transformation that maps any orthonormal basis to any other orthonormal basis.}  Then define
\bea
\nonumber
{\cal U}_{BC} & = & \sum_{\ell = 1}^L \identity_A \otimes u_B^{(\ell)} \otimes \vert \ell \rangle_C {}_C \langle \ell \vert \\
\label{block}
& = & \identity_A \otimes \left(\begin{array}{ccc} u_B^{(1)} & & \\ & \ddots & \\ & & u_B^{(\ell)} \end{array}\right)_{BC}
\eea
This is the controlled unitary we referred to previously.  It has the property that
\be
\label{factorized}
{\cal U}_{BC} \vert \psi \rangle = \sum_{\ell = 1}^L \sqrt{p_\ell} \, \vert \ell \rangle_A \otimes \vert \chi \rangle_B \otimes \vert \ell \rangle_C
\ee
A controlled unitary for the Werner state is given in appendix \ref{appendix:Werner}.

The key point is that the right hand side of (\ref{factorized}) is a tensor product.  With a slight abuse of notation we can write it as
\be
{\cal U}_{BC} \vert \psi \rangle = \vert \chi \rangle_B \otimes \Big(\sum_{\ell = 1}^L \sqrt{p_\ell} \, \vert \ell \rangle_A \otimes \vert \ell \rangle_C\Big)
\ee
The unitary ${\cal U}_{BC}$ shifts all of the entanglement of $A$ with the combined $BC$ system so that it becomes entanglement just between $A$ and $C$.  As shown in appendix \ref{appendix:UBC},
the entanglement groups for the state ${\cal U}_{BC} \vert \psi \rangle$ reflect the fact that system $B$ has been completely disentangled.

Note that we can go on to construct a controlled unitary ${\cal U}_{AC}$ which completely disentangles system $A$.  By combining the two transformations, we can map the original state to a tensor product.
\be
{\cal U}_{AC} \, {\cal U}_{BC} \vert \psi \rangle = \vert \chi \rangle_A \otimes \vert \chi \rangle_B \otimes \Big(\sum_{\ell = 1}^L \sqrt{p_\ell} \, \vert \ell \rangle_C \Big)
\label{sepprep}
\ee
This characterization of pure states whose reduced density matrices are separable is related to the usual notion of how one makes a separable
density matrix.  From (\ref{sepprep}) there is a controlled unitary ${\cal U}_{ABC}={\cal U}^{-1}_{AC}{\cal U}^{-1}_{BC} $ with party $C$ as the control,
for which
\be
{\cal U}_{ABC} \, \vert \, \hbox{\rm tensor product state} \, \rangle = \vert \psi \rangle
\label{dmp}
\ee
The usual  procedure for making a separable $\rho_{AB}$ is for party $C$ to generate a random number $\ell \in \lbrace 1,\ldots,L \rbrace$  with probability $p_\ell$, and to
communicate the result to parties $A$ and $B$ who then prepare their states accordingly.  Party $C$ could generate the random number by making
a measurement of $\ell$ in the state $\sum_{\ell = 1}^L \sqrt{p_\ell} \, \vert \ell \rangle_C$.  Due to the delayed measurement principle \cite{aharonov1998quantumcircuitsmixedstates,GUREVICH202221}, an alternate procedure is to apply a controlled unitary, then at the end measure
(or trace over) system $C$.  This results in identical statistical outcomes for
the $AB$ system. This alternate procedure can be implemented by applying ${\cal U}_{ABC}$ as in (\ref{dmp}) then measuring or tracing over $C$.

%%%%%%%%%%%%%%%%%%%%%%%%%%%%%%%%%
\section{Relation to quantum tasks \label{section:qtasks}}
%%%%%%%%%%%%%%%%%%%%%%%%%%%%%%%%%
In this section we analyze a few well-known quantum tasks and show that the definition of entanglement presented in section \ref{sect:groups} is at the heart of what makes the task possible. For tasks relying on two-party entanglement we will see that the tasks succeed because the entangled states being used are eigenstates of either the operators  $Z_1 \otimes Z_2$ and  $X_1 \otimes X_2$ (when $E_{AB} = PSU(2)$) or just $Z_1 \otimes Z_2$ (when $E_{AB} = U(1)$).

In this section we use the following notation.  For qubits the operators $X,Y,Z$ are
\begin{equation}
X=
\begin{pmatrix}
0 &1  \\
1 & 0
\end{pmatrix}
\ \ \ \ \ Y=
\begin{pmatrix}
0 & -i  \\
i & 0
\end{pmatrix}
\ \ \ \ \ Z=
\begin{pmatrix}
1 &0  \\
0 & -1
\end{pmatrix}
\end{equation}
$X_{i},Y_i,Z_{i}$ indicate these operators acting on the $i$${}^{th}$ qubit. The states $\vert 0 \rangle$ and $\vert 1 \rangle$
are eigenstates of the $Z$ operator with eigenvalue $+1$ and $-1$, respectively, and the state $|*\rangle_{i}$ indicates such
states for the $i$${}^{th}$ particle. 

%%%%%%%%%%%%%%%%%%%%%%%%%%%%%%%
\subsection{Super-dense coding\label{sect:superdense}}
%%%%%%%%%%%%%%%%%%%%%%%%%%%%%%%
Given a single qubit we can only encode one bit of classical information, by choosing two orthogonal states in the Hilbert
space to represent the classical bit.  Super-dense coding (in its simplest form) is the ability to send one qubit to a receiver
who can then extract two bits of classical information.

Super-dense coding \cite{PhysRevLett.69.2881} works as follows. Alice and Bob share a pair of qubits in a Bell state. Alice acts on her qubit with whatever local unitary
she wants and then sends her qubit to Bob. This enables her to send two classical bits of information, since she can act on her qubit to produce four orthogonal states
of the entangled pair which Bob can then measure. How is it that Alice is able to act on her qubit to produce four orthogonal states?
Acting on a single isolated qubit can only produce two orthogonal states, so to get four orthogonal states of two qubits you might think you need to act on both qubits.

Indeed starting from a general (up to local unitaries) two qubit state $\ket{\psi} = a\ket{00} +b \ket{11}$ one can act with local unitary operators on the two systems to produce four orthogonal states.
\bea
\nonumber
&& \big( {\bf I} \otimes {\bf I} \big) \, \ket{\psi} = a \ket{00} +b \ket{11} \\
\nonumber
&& \big( X \otimes {\bf I} \big) \, \ket{\psi} = a \ket{10} +b \ket{01} \\
&& \big( {\bf I} \otimes Y \big) \, \ket{\psi} = ia \ket{01} - ib \ket{10}\\
\nonumber
&& \big( X \otimes Y \big) \, \ket{\psi} = ia \ket{11} - ib \ket{00}
\eea
But thanks to the maximal entanglement of the
Bell state, any unitary acting on the second qubit is equivalent to a unitary acting on the first.  This means that starting from a Bell state we can produce four orthogonal
states by only acting on the first qubit with
\be
\label{denseops}
{\bf I} \otimes {\bf I},\  X \otimes {\bf I}, \  Y  \otimes {\bf I} , \ XY \otimes {\bf I}
\ee
Thus super-dense coding is directly related to the characterization of entanglement we have described.

%%%%%%%%%%%%%%%%%%%%%%%
\subsection{Teleportation}
%%%%%%%%%%%%%%%%%%%%%%%
Before getting to teleportation let us do the following exercise. We start with two qubits, one in a state $\ket{\phi}=a|0\rangle+b|1 \rangle$ and the other in a state $\ket{0}$.
We act on both with the CNOT operation where the first qubit is the control qubit.
That is, we act with the unitary operator
\begin{equation}
CNOT_{12}=\frac{1}{2}(Z_{1}\otimes I_{2}+I_{1}\otimes I_{2}+I_{1}\otimes X_{2}-Z_{1}\otimes X_{2})
\label{cnot12}
\end{equation}
We then act on the first qubit with the Hadamard operation
\be
\label{Hadamard}
H = \frac{1}{\sqrt{2}}(Z+X)
\ee
to get the state
\begin{equation}
|0 \rangle \otimes (a|0 \rangle+b|1 \rangle)+ |1 \rangle \otimes (a|0 \rangle-b|1 \rangle)
\end{equation}
Suppose we measure $Z_1$.  If we get $+1$ then qubit 2 is in the state $\ket{\phi}$ and if we get $-1$ then we act with $Z_2$ on qubit 2 and again put it in the state $\ket{\phi}$.
We can add a third spectator qubit without changing anything. There is nothing mysterious about this result. 

Now suppose one starts with qubit $1$ in a state $\ket{\phi}=a\ket{0}+b\ket{1}$ and qubits $2$ and $3$ in a Bell state, say $\frac{1}{\sqrt{2}}(|00 \rangle+|11 \rangle)$.
One acts on qubits $1$ and $2$ with the $CNOT_{12}$ gate (\ref{cnot12}) and then with the Hadamard gate on qubit $1$. This gives the state
\begin{equation}
\frac{a}{\sqrt{2}}(|0 \rangle+|1 \rangle)\otimes(|00 \rangle+|11 \rangle)+\frac{b}{\sqrt{2}}(|0 \rangle-|1 \rangle)\otimes(|10 \rangle+|01 \rangle)
\end{equation}
One can now measure $Z_1$ and $Z_3$ and get qubit $2$ in some state. The four possible results are
\begin{eqnarray}
00 &\rightarrow&  a\ket{0}+b\ket{1} \nonumber \\
01  &\rightarrow&  a\ket{1}+b\ket{0} \nonumber\\
10   &\rightarrow&  a\ket{0}-b\ket{1} \nonumber\\
11  &\rightarrow&  a\ket{1}-b\ket{0}
\label{teleres}
\end{eqnarray}
Knowing the answer of the measurements one can act with the unitaries $(I_2, X_{2}, Z_{2}, Z_{2}X_{2})$ respectively to get qubit $2$ into the desired state $\ket{\phi}$.
Again there is no mystery here, it is just a more elaborate example than the first one. Teleportation  \cite{PhysRevLett.70.1895} happens when we instead
measure $Z_1$ and $Z_2$ and then look at the state of qubit $3$. We get the same result as (\ref{teleres}), but for qubit $3$. That is, by manipulating qubits $1$ and $2$
one teleports the state of qubit $1$ to a seemingly un-manipulated qubit $3$. Why does this work? There does not seem to be any symmetry between qubit $2$ and qubit $3$
since the $CNOT_{12}$ gate acted on qubits $1$ and $2$ and not on qubit $3$. However one could define a $CNOT_{13}$ gate acting on qubits $1$ and $3$ by
\begin{equation}
\frac{1}{2}(Z_{1}\otimes I_{3}+I_{1}\otimes I_{3}+I_{3}\otimes X_{3}-Z_{1}\otimes X_{3}).
\label{cnot13}
\end{equation}
If we follow the same procedure as above and measure  $Z_1$ and $Z_2$, we would not be surprised that qubit $3$ is found in the state (\ref{teleres}). Now comes the crucial part.
\begin{equation}
CNOT_{12}\Big(  \frac{1}{\sqrt{2}}(a|0 \rangle+b|1 \rangle)\otimes(|00 \rangle+|11 \rangle)\Big)=CNOT_{13}\Big(  \frac{1}{\sqrt{2}}(a|0 \rangle+b|1 \rangle)\otimes(|00 \rangle+|11 \rangle)\Big)
\end{equation}
The reason this works is that on the Bell state we started with, the action of $X_2$ is equivalent to the action of $X_3$, due exactly to
the property of entanglement we have defined.  So we see that teleportation works exactly because of the property that is captured
by our definition of entanglement.

%%%%%%%%%%%%%%%%%%%%%%%
\subsection{Entanglement swapping}
%%%%%%%%%%%%%%%%%%%%%%%
To establish notation, the Bell states are listed in the following table.  In the right columns we list their eigenvalues under certain operators which stabilize the state.
\be
\begin{array}{c|c|c}
{\rm state} & X_1 \otimes X_2 & Z_1 \otimes Z_2 \\
\hline
\vert \phi_1\rangle = \frac{1}{\sqrt{2}}(|00\rangle+|11\rangle) & + & + \\[8pt]
\vert \phi_2\rangle =\frac{1}{\sqrt{2}}(|00\rangle-|11\rangle) & - & + \\[8pt]
\vert \phi_3\rangle = \frac{1}{\sqrt{2}}(|01\rangle+|10\rangle) & + & - \\[8pt]
\vert \phi_4\rangle = \frac{1}{\sqrt{2}}(|01\rangle-|10\rangle) & - & -
\end{array}
\ee
If qubit $i$ and qubit $j$ are in Bell state $\ket{\phi_k}$ we write it as $\ket{\phi_k}_{ij}$.

Say one has four qubits in a state $\ket{\phi_1}_{12}\otimes \ket{\phi_1}_{34}$. Entanglement swapping is the observation that if one projects qubits $2$ and $3$ into
the Bell state  $\ket{\phi_k}_{23}$ then one finds qubits $1$ and $4$ in the state $\ket{\phi_k}_{14}$. As we shall see, this phenomenon can be understood in terms
of local unitary stabilizers which in our perspective define the entanglement of the state.

Note that the state $\ket{\phi_1}_{12}\otimes \ket{\phi_1}_{34}$ can be expanded in a tensor product of Bell states, just like any other state in the Hilbert space.
\begin{equation}
|\phi_1 \rangle_{12}\otimes |\phi_1 \rangle_{34}=\sum_{k,l=1}^{4} a_{kl} \, |\phi_k \rangle_{14}\otimes |\phi_l \rangle_{23}.
\label{entswap1}
\end{equation}
The left side is an eigenstate with eigenvalue $+1$ under the operators
\begin{equation}
\label{swap}
\big(X_1\otimes X_{2}\big) \otimes \big(X_3 \otimes X_4\big), \qquad \big(Z_{1}\otimes Z_{2}\big) \otimes \big(Z_3 \otimes Z_{4}\big)
\end{equation}
In this form the invariance reflects some of the $12$ and $34$ entanglement of the state.

The right side of (\ref{entswap1}) must have the same property.  Every state appearing on the right side of (\ref{entswap1}) is an eigenstate
of these operators with eigenvalue $\pm 1$, where now we are thinking of the operators as reflecting some of the $14$ and $23$
entanglement of the state.  It is straightforward to check that the only combinations with eigenvalue $+1$ under both operators
have $k = l$, thus the sum in (\ref{entswap1}) only includes terms with $k=l$.  Some further calculation
shows that $a_{kl}=\pm \frac{1}{2} \delta_{kl}$, but this fact is not essential for entanglement swapping.

%%%%%%%%%%%%%%%%%%%%%%%
\subsection{CHSH inequalities}
%%%%%%%%%%%%%%%%%%%%%%%
A famous consequence of bipartite entanglement is violation of the Bell inequalities.  Here we show that violation of the CHSH
inequalities \cite{PhysRevLett.23.880} (generalizations of the Bell inequalities) is related to entanglement as we have defined it.
We will give the discussion in the context of two qubits. 

Let us look at   two statistical variables  of system $1$, which can take values $\pm 1$, and label them by  $a, a'$. Similarly for system
$2$ we have variables labeled by $b, b'$.
Then under the assumption of local hidden variable models (LHVM) one gets the CHSH inequality 
\begin{equation}
\big|\langle ab+ab'+a'b-a'b' \rangle\big|\leq 2
\label{chshcombo}
\end{equation}
This inequality is known to be violated in quantum mechanics for any two-qubit entangled state \cite{GISIN1991201}. That is, for any two qubit state $|\psi \rangle$
which is not a tensor product state, there is some choice of observables $a,a',b,b'$ for which the left hand side of (\ref{chshcombo}) is larger than 2.

Let us see how this is connected with the view of entanglement we have presented.  Any state of two qubits can be written in a Schmidt basis as
\be
\label{CHSHSchmidt}
\ket{\psi} = \sum_{i = 1}^2 p_i \, \ket{i}_1 \otimes \ket{i}_2
\ee
The tensor product state has $p_1 = 1$, $p_2 = 0$.  All other states have $p_1 > p_2 > 0$, and as a result they have
at least a non-trivial $U(1)$ stabilizer group which acts as
\be
u_1(\theta) \otimes u_2(\theta) = \left(\begin{array}{cc} e^{i \theta} & 0 \\ 0 & e^{-i \theta} \end{array}\right) \otimes
\left(\begin{array}{cc} e^{-i \theta} & 0 \\ 0 & e^{i \theta} \end{array}\right)
\ee
We can build Hermitian operators $h_1,h_2$ from the stabilizer by setting
\be
h_1 = -i u_1\big({\pi \over 2}\big) = Z_1, \quad h_2 = i u_1\big({\pi \over 2}\big) = Z_2
\ee
The entanglement of the state is partially encoded in the fact that $\big( Z_1 \otimes Z_2 \big) \ket{\psi} = \ket{\psi}$.

We need to construct observables $a,a',b,b'$ with eigenvalues $\pm 1$ to use in the CHSH inequality.  Besides $Z_1$ and $Z_2$ we will use $X_1$ and $X_2$ since
these operators have the helpful properties that they anti-commute and square to the identity.  With these ingredients we can construct suitable observables ($\epsilon$
is a small parameter)
\begin{equation}
a=Z_{1}, \ \  a'=X_{1},\ \ b=\frac{Z_{2}+\epsilon X_{2}}{\sqrt{1+\epsilon^2}},\ \ b'=\frac{Z_{2}-\epsilon X_{2}}{\sqrt{1+\epsilon^2}}.
\end{equation}
With these operators the left side of the CHSH inequality (\ref{chshcombo}) becomes
\begin{equation}
{2 \over \sqrt{1 + \epsilon^2}} \, \big\vert \big\langle \, Z_1 Z_2 + \epsilon X_1 X_2 \, \big\rangle \big\vert
\end{equation}
Expanding to first order in $\epsilon$ this is $2 \vert \langle Z_{1}Z_{2} \rangle + \epsilon \langle X_{1}X_{2} \rangle \vert$. Crucially, due to the
entanglement of the state, we have $\langle Z_{1} Z_{2} \rangle=1$, and an explicit computation with (\ref{CHSHSchmidt}) gives
$\langle X_1 X_2 \rangle = 2 p_1 p_2 > 0$.  So the left side of (\ref{chshcombo}) becomes $\vert 2 + 4 \epsilon p_1 p_2 \vert$ which is larger than 2
for positive $\epsilon$. 

Note the crucial role of the non-trivial entanglement group: not only do we need $Z_{1} Z_{2} |\psi \rangle=|\psi \rangle$, but we also need that $Z_{i}$ by itself does not stabilize the state,
$Z_{i} |\psi \rangle \neq \pm |\psi \rangle$.  If say $Z_1$ stabilized the state, we would have $\langle X_1 X_2 \rangle = \langle Z_1 X_1 X_2 Z_1 \rangle = - \langle X_1 X_2 \rangle$
since $X_1$ and $Z_1$ anti-commute.  Then $\langle X_1 X_2 \rangle = 0$, and the CHSH inequality would not be violated.  The connection between maximum violation of Bell-like
inequalities and stabilizers was exhibited for graph states in \cite{G_hne_2005}.

%%%%%%%%%%%%%%%%%%
\subsection{Simons algorithm}
%%%%%%%%%%%%%%%%%%
We now turn to examples that rely on multi-partite entanglement. In this subsection we discuss the Simons algorithm, and in the next subsection we will discuss
a quantum pseudo-telepathy game.  For the Simons algorithm, let $B=\lbrace \, 0,1\, \rbrace$, so that $B^n$ is the set of strings of length $n$ with each entry being $0$ or $1$.
On $B^n$ we introduce an operation of addition: if $x,y \in B^n$ we denote $(x \oplus y) \in B^n$ obtained by addition modulo $2$. We also denote $x \cdot y= x_1 y_1 \oplus \cdots \oplus x_n y_n \in B$.
 
 Suppose we are given a function $f \, : \, B^n \rightarrow B^n$. 
 We are told the function is $2 \, : \, 1$ 
and that it has the property that for all $x,y \in B^n$
\begin{equation}
f(x)=f(y) \ \ \hbox{\rm if and only if} \ \ y=x \oplus \xi
\end{equation}
Here  $\xi$ is some fixed element of $B^n$.  The task is to determine $\xi$. 

One can of course evaluate $f(x)$ for all $x$ and find $\xi$ by an exhaustive search, but this requires exponential time in $n$.
The Simons algorithm \cite{doi:10.1137/S0097539796298637} provides a way to find $\xi$ in polynomial time (for a nice description see also \cite{Jozsa:1996hb}).
It starts with having  a black box which implements on any two strings of $n$ qubits (say $|x \rangle$ and $|y \rangle$),  a unitary transformation 
\begin{equation}
U_{f}\  : |x \rangle \otimes |y \rangle \rightarrow |x \rangle \otimes |y\oplus f(x) \rangle.
\end{equation}
One then starts with two systems of $n$ qubits in the state $|0\cdots 0\rangle \otimes |0\cdots 0 \rangle$ and applies the local unitary $H^{\otimes n}$ to the first system
to get $\sum_{x\in B^n}|x\rangle \otimes |0\cdots 0\rangle$. Here $H$ is the Hadamard gate
(\ref{Hadamard}); since $H \ket{0} = \ket{+} = {1 \over \sqrt{2}} \big(\ket{0} + \ket{1}\big)$ the state of the first system is just $|+ \cdots + \rangle$. One now acts with $U_{f}$ and gets
\begin{equation}
\sum_{x \in B^n}|x \rangle \otimes |f(x) \rangle
\end{equation}
Now there is bipartite entanglement between the two systems. In fact since 
\begin{equation}
\sum_{x\in B^n} |x \rangle \otimes |f(x) \rangle = \sum_{f(x)} \big(|x \rangle +|x\oplus \xi\rangle\big)\otimes |f(x) \rangle
\end{equation}
the Schmidt decomposition is such that half of the Schmidt coefficients are $1$ and half are $0$.  This leads to a large two-party
entanglement group for this state. But this entanglement is only partly responsible for the ability to extract $\xi$ from the state in
polynomial time.

The procedure is to measure the second system, which collapses the first system into the state
\be
\label{SimonsForm}
|x_{0}\rangle + |x_{0} \oplus \xi \rangle
\ee
for some $x_{0}$. But how to determine $\xi$?  At this stage measuring the state of the first system in the $Z$ basis will not give any
information since $x_{0}$ could be any element of $B^n$.  Repeating the whole procedure from the beginning will most likely
produce a different $x_{0}$.  We need a procedure to extract the periodicity $\xi$ in a way that is independent of $x_{0}$.

To decipher $\xi$ one can  use the fact that states of the form (\ref{SimonsForm}) have a local unitary stabilizer (corresponding to a non-trivial multi-partite entanglement group ${\mathbb Z}_{2}$) which depends on $\xi$ but is
independent of $x_{0}$. To see this note that $\xi$ is a string of $0$'s and $1$'s. The difference between $|x \rangle$ and
$|x \oplus \xi \rangle$ is only in the places where $\xi$ has $1$'s, and in these places the difference is $0 \leftrightarrow 1$.
This means $|x_{0}\rangle + |x_{0}\oplus \xi \rangle$ is invariant under
\begin{equation}
{\bf I} \otimes \cdots \otimes {\bf I} \otimes X_{i_1} \otimes {\bf I} \otimes \cdots \otimes {\bf I} \otimes X_{i_{k}} \otimes {\bf I} \otimes \cdots
\otimes {\bf I}
\label{simonsymx}
\end{equation}
where $i_{1},\ldots,i_k$ are the $k$ locations where $\xi$ has $1$'s.  Since (\ref{simonsymx}) is a local unitary transformation it encodes
entanglement of the qubits. In fact the state (\ref{SimonsForm}) is a tensor product of $k$ entangled qubits -- the qubits where $\xi$ had $1$'s,
which are in a $k$-party GHZ state -- with the other $n - k$ qubits which are in an unknown tensor product state. This implies
that measurements will have correlated outcomes: if we measure the state in the $X$ basis we will always find 
\be
X_{i_1} \cdots X_{i_k} = 1, 
\label{simcor}
\ee
while other strings of $X$'s are un-correlated. To see what this means imagine we measure the $X_{i}$ of each qubit getting some string of $\pm 1$'s.
Transferring this to bit symbols $(1\rightarrow 0)$ and $( -1 \rightarrow 1)$, we get an object $x \in B^n$. Any $x$ obtained in this way satisfies $x \cdot \xi =0$:
places where $\xi$ has a $0$ make no contribution, while places where $\xi$ has a $1$ gives a sum that vanishes mod 2 since (\ref{simcor}) means that
measuring $X_{i_1},\ldots,X_{i_k}$ gives an even number of $-1$'s. So measuring in the $X$ basis gives a vector satisfying $x \cdot \xi = 0$, and by repeating the
procedure enough times one can determine $\xi$.  Note that $x \cdot \xi = 0$ is entirely due to the multipartite entanglement of the state (\ref{SimonsForm}).

In the literature one usually acts with $H^{\otimes n}$ before measuring.  This maps the $X$ basis to the $Z$ basis and gives\footnote{$H^{\otimes n}$ is a local unitary transformation so it does not change the entanglement.}
\begin{equation}
H^{\otimes n} \big(|x_{0}\rangle + |x_{0} \oplus \xi \rangle\big) = \sum_{\scriptstyle z \in B^n \atop \scriptstyle z \cdot \xi =0}(-1)^{x_{0}\cdot z} \, |z\rangle.
\end{equation}
Every state $\ket{z}$ appearing in the sum is invariant under the local unitary transformation (\ref{simonsymx}) with $X\rightarrow Z$.
Measuring the state in the $Z$ basis gives a vector $z$ satisfying $z \cdot \xi = 0$.  One can show that repeating the procedure $O(n^2)$ times gives a sufficiently large
collection of vectors satisfying $z \cdot \xi = 0$ in order to determine $\xi$. Note that one does not actually need to do the Fourier transform to get $\xi$, as one could measure
in the $X_{i}$ basis rather than the $Z_{i}$ basis.  Equivalently in this algorithm the Fourier transform is a local unitary transformation (unlike in say the Shor algorithm).

%%%%%%%%%%%%%%%%%%
\subsection{Quantum pseudo-telepathy \label{subsection:qpt}}
%%%%%%%%%%%%%%%%%%
Here we describe an example of a quantum pseudo-telepathy game (for a nice explanation see \cite{brassard2004recasting}). These are tasks (games)  that no classical strategy can win with certainty, but there is a quantum strategy to win if one uses entangled states. An understanding of which underlying property of the state controls the probability of winning probability illustrates our claim that entanglement should be
understood in terms of local unitary stabilizers.

%%%%%%%%%%%%%%%%%%
\subsubsection{Three party GHZ game}
%%%%%%%%%%%%%%%%%%
The three party GHZ game is as follows \cite{Greenberger1989, mermin1990quantum}. There are 3 parties $A,B,C$. 
The three parties ($i=1,2,3$) get as input $q_{i} \in \{0,1 \}$. The input is not general, but obeys $\sum_{i=1}^{3} q_{i}=0 \ {\rm mod} \ 2$. In response each party must produce an output $p_{i} \in \{0,1\}$. The game is won if the output obeys
\bea
\sum_{i=1}^{3} p_{i} &=& 0 \ {\rm mod} \ 2 \ {\rm if} \  q_{i}=(0,0,0) \label{1}\\
\sum_{i=1}^{3} p_{i} &=& 1\ {\rm mod}\ 2\ {\rm if} \  q_{i}=(0,1,1)\label{2}\\
\sum_{i=1}^{3} p_{i} &=& 1\ {\rm mod}\ 2\ {\rm if} \ q_{i}=(1,1,0)\label{3}\\
\sum_{i=1}^{3} p_{i} &=& 1\ {\rm mod}\ 2\ {\rm if} \ q_{i}=(1,0,1)\label{4}
\label{wintab}
\eea
The best classical strategies give  probability of wining $0.75$.

A winning quantum strategy is as follows. One starts with the GHZ state 
\be
\frac{1}{\sqrt{2}}\big(\vert 000 \rangle + \vert 111 \rangle\big)
\ee
 distributed to the three parties. Each of the parties which gets an input $0$ measures the observable $X$ of its qubit,  and each one that gets a $1$ measures the observable $Y$. The output is $0$ if the measurement gives $+1$ and the output is $1$ if the measurement gives $-1$. Let us see that this strategy wins with certainty.  The GHZ state has the property that 
\be
X_1 X_2 X_3 \vert {\rm GHZ} \rangle = \vert {\rm GHZ} \rangle
\label{gx}
\ee
Thus if everyone measures $X$ there will be an even number of $-1$, so there will be an even number of $1$ in the output, satisfying the first line in the winning condition. Since 
\be
Z_1 Z_2 \vert {\rm GHZ} \rangle=\vert {\rm GHZ} \rangle,
\ee
 this gives together with (\ref{gx}) the result 
 \be
 Y_1 Y_2 X_3 \vert {\rm GHZ} \rangle=-\vert {\rm GHZ} \rangle.
 \ee
Thus measuring $Y_1, Y_2, X_3$ will give an odd number of $-1$ and thus as output an odd number of $1$'s. This satisfies the third line in the winning table.
Since
\be
Z_3 Z_2 \vert {\rm GHZ} \rangle=\vert {\rm GHZ} \rangle\  {\rm and} \  Z_1 Z_3 \vert {\rm GHZ} \rangle=\vert {\rm GHZ} \rangle,
\ee
we also satisfy the second and fourth lines in the winning table.  To summarize, the relationships that lead to winning with certainty are
\bea
{\rm Success\ of}\ (\ref{1}) &\leftrightarrow& X_1X_2 X_3 \vert {\rm GHZ} \rangle=\vert {\rm GHZ} \rangle\label{1a}\\
{\rm Success\ of}\ (\ref{2}) &\leftrightarrow& X_1Y_2 Y_3 \vert {\rm GHZ} \rangle=-\vert {\rm GHZ} \rangle\label{2a}\\
{\rm Success\ of}\ (\ref{3}) &\leftrightarrow& Y_1Y_2 X_3 \vert {\rm GHZ} \rangle=-\vert {\rm GHZ} \rangle\label{3a}\\
{\rm Success\ of}\ (\ref{4}) &\leftrightarrow&Y_1X_2 Y_3 \vert {\rm GHZ} \rangle=-\vert {\rm GHZ} \rangle\label{4a}
\eea
This shows that winning is controlled by  the local unitary stabilizers of the state.

%%%%%%%%%%%%%%%%%%%%%
\subsubsection{Breaking the symmetries\label{subsub:GHZ}}
%%%%%%%%%%%%%%%%%%%%%
To make this even clearer, let us look at what happens when we use in the game, a state that does not have the necessary local stabilizers.
Say we fix a state $\vert\phi\rangle$ which we use for the GHZ game. Each of the unitary operators in (\ref{1a}) -- (\ref{4a}) has eigenvalues $\pm1$. Thus for each unitary we can expand
\be
\vert\phi\rangle =\sum_{i}a_{i} \vert+1_{i}\rangle+\sum_{j}b_{j} \vert-1_{j}\rangle
\ee
where $\vert +1_{i}\rangle$ and $\vert -1_{j}\rangle$ form an orthonormal basis for the subspaces with eigenvalues $\pm 1$, for any particular operator appearing in (\ref{1a}-\ref{4a}),
and the coefficients $a_{i}$ and $b_j$ depend on which operator we use to form the basis.
The probability of failure for case (\ref{1}) is the appropriate $\sum_{j} |b_{j}|^{2}$, while the probability of failure for cases (\ref{2}) -- (\ref{4}) is the appropriate $\sum_{i} |a_{i}|^{2}$.
Let us define four order parameters (only for a state with the right stabilizers will the appropriate order parameters vanish)
\beas
\xi_1 &=& \frac{1}{2}\big(1-\langle X_1X_2 X_3\rangle\big)\\
 \xi_{2} &=& \frac{1}{2}\big(1+\langle X_1Y_2 Y_3\rangle\big)\\
 \xi_{3} &=& \frac{1}{2}\big(1+\langle Y_1Y_2 X_3\rangle\big)\\
 \xi_{4} &=& \frac{1}{2}\big(1+\langle Y_1X_2 Y_3\rangle\big)
\eeas
It is now easy to check that the probability of failure in each of the four possibilities is exactly the order parameter associated with that operator.
Thus given a state $\vert \phi \rangle$, the probability of failure is simply given by (assuming each input has probability $1/4$)
\be
P_{\rm failure}=\frac{1}{4}\sum_{i=1}^{4} \xi_{i}
\ee
This connects the probability of failure directly with the order parameters, which measure how close the state $\vert \phi \rangle$ is to having the necessary stabilizers.

%%%%%%%%%%%%%%%%%%%%%%%%%%%%
\subsubsection{General GHZ game}
%%%%%%%%%%%%%%%%%%%%%%%%%%%%
In the general GHZ game  \cite{PhysRevLett.65.1838, Buhrman_2003, boyer2004extendedghznplayergames}, there are $N$ players each receiving as input a number
$q_{i} \in \{0,1,\ldots, D-1\}$.  The input data is not arbitrary but obeys
\be
{1 \over D} \sum_{i=1}^{N} q_{i} = m
\label{defn}
\ee 
where $m$ is an integer. Each player after receiving the input gives an output $p_{i} \in \{0,1\}$. The game is won if
\bea
& &\sum_{i}^{N} p_{i}=0 \ {\rm mod} \ 2\  \leftrightarrow \ \hbox{\rm $m$ even} \nonumber\\
& &  \ \sum_{i}^{N} p_{i}=1 \ {\rm mod} \ 2\  \leftrightarrow \ \hbox{\rm $m$ odd}
\label{win}
\eea
As before, the players are allowed to devise any strategy before they get the input, but they cannot communicate after receiving the input. If $D$ is even, there is no classical strategy that wins the game with probability 1. For $D=2$ the probability of winning with the best classical strategy is $\frac{1}{2}+\frac{1}{2^{[N/2]}}$ .

A winning quantum strategy is as follows. Before the game an $N$-partite GHZ state
\be
\vert {\rm GHZ} \rangle=\frac{1}{\sqrt{2}}\big(\vert 0^N \rangle + \vert 1^N \rangle\big)
\ee
is distributed between the parties. Given an input $q_{i}$, we define $\theta_{i}= \frac{\pi q_{i}}{D}$. Then player $i$ makes a measurement of the observable 
\be
X_{\theta_{i}}=X_{i} \cos \theta_{i} +Y_{i} \sin \theta_{i} 
\ee
If the result is $1$ the output is $p_{i}=0$, and if the result is $-1$ the output is $p_{i}=1$.

To see why this strategy works, note that the $N$-party GHZ state has local unitary stabilizers.  Rather obviously, $\Pi_{i=1}^{N} X_{i}$ is a stabilizer. Another stabilizer is $\Pi_{i=1}^{N} U_{i}$ with 
\be
U_{i} = \left(\begin{array}{cc} e^{i\theta_{i}} & 0 \\ 0 & e^{-i\theta_{i}} \end{array}\right)
\ee
This is a stabilizer of the $N$-party GHZ state if
\be
\sum_{i=1}^{N} \theta_{i}=m \pi
\label{defm}
\ee
for any integer $m$.  If $m$ is even then
\be
\Pi_{i=1}^{N} U_{i} \, \vert {\rm GHZ} \rangle=\vert {\rm GHZ} \rangle
\ee
and if $m$ is odd
\be
\Pi_{i=1}^{N} U_{i} \, \vert {\rm GHZ} \rangle=-\vert {\rm GHZ} \rangle.
\ee
Note that all such $\Pi_{i=1}^{N} U_{i}$ are made of products of two-party stabilizers.  Indeed the GHZ state has $U(1)$ two-party entanglement groups for any two parties, as well as an $N$-party
entanglement group which is ${\amsbb Z}_2$.

Now note that
\be
(\Pi_{i=1}^{N} X_{i})( \Pi_{i=1}^{N} U_{i})=\Pi_{i=1}^{N} X_{\theta_{i}} 
\ee
This means
\be
\Pi_{i=1}^{N} X_{\theta_{i}}\vert {\rm GHZ} \rangle=\pm\vert {\rm GHZ} \rangle
\label{usefulsym}
\ee
where the $+$ sign is for $m$ even and the $-$ sign is for $m$ odd (see (\ref{defm})). Since the input data obeys (\ref{defn}) and using the definition of $\theta_{i}$, we see that given (\ref{defm}), upon measuring $X_{\theta{i}}$, the outcome must obey (\ref{usefulsym}).  So for $m$ even the output will contain even number of $1$'s and for odd $m$ it will contain an odd number of $1$'s.  This satisfies the
winning condition (\ref{win}). 

Another way of saying this is as follows. Define
\be
\widetilde{X}_{\theta_{i}}=\frac{1}{2}(I-X_{\theta_{i}})
\ee
which has eigenvalues $\in \{0,1\}$. Note that
\be
e^{i\pi\sum_{i=1}^{N} \tilde{X}_{\theta{i}}}= \Pi_{i=1}^{N} X_{\theta{i}}
\ee
which means
\be
e^{i\pi\sum_{i=1}^{N} \tilde{X}_{\theta{i}}} \, \vert {\rm GHZ} \rangle= (-1)^m \vert {\rm GHZ} \rangle
\ee
where $m$ is defined in (\ref{defm}). This means that when measuring $\sum_{i=1}^{N} \tilde{X}_{\theta{i}}$ one would get an even integer if $m$ is even and an odd integer if $m$ is odd, thereby winning
the game.

This analysis makes it clear that local unitary stabilizers underlie the wining strategy, and should thus be thought of as capturing the entanglement of the GHZ state.

%%%%%%%%%%%%%%%%%%%%%%%%%%%%
\subsubsection{Breaking the symmetries\label{subsub:generalGHZ}}
%%%%%%%%%%%%%%%%%%%%%%%%%%%%
Suppose we have an $N$-party state $\vert \phi \rangle$ which we use for the general GHZ game, but which is not necessarily the $N$-party GHZ state. We have a set of possible inputs labeled by $J$,
$\{q_{i}\}_{J} \in \{0,\cdots D-1\}^{N}$. For each such input there is a set of local unitary operators we used, namely
\be
{\cal X}_{J}=\Pi_{i=1}^{N} X_{\theta{i}}.
\label{opp}
\ee
These operators obey ${\cal X}_{J}^{2}=1$, so they have eigenvalues $\pm1$.  We can expand $\vert \phi \rangle$ in eigenvectors,
\be
\vert \phi \rangle = \sum_{i} a_{i} \vert +1_{i} \rangle + \sum_{l} b_{l} \vert -1_{l} \rangle
\ee
where $\vert +1_{i} \rangle$ and $\vert -1_{l} \rangle$ form an orthonormal basis for the subspaces with eigenvalues $\pm 1$, for any particular operator ${\cal X}_{J}$ appearing in (\ref{opp}).
The $a_{i}$'s and $b_l$'s depend of course on which operator we use to expand the state.
The probability of failure for even $m$ is the appropriate $\sum_{l} |b_{l}|^{2}$, while the probability of failure for odd $m$ is the appropriate $\sum_{i} |a_{i}|^{2}$.
We define order parameters which vanish only in states with the appropriate stabilizers
\bea
\nonumber
\xi_J &=& \frac{1}{2}(1-<{\cal X}_{J}>) \ \ \hbox{\rm for even $m$} \\
\xi_J &=& \frac{1}{2}(1+<{\cal X}_{J}>) \ \ \hbox{\rm for odd $m$}
\eea
It is easy to check that the probability of failure for each input $\{q_{i}\}_{J}$ is exactly the order parameter associated with the corresponding operator.
Thus given a state $\vert \phi \rangle$, the probability of failure is given by (assuming each input $\{q_{i}\}_{J}$ has equal probability)
\be
P_{\rm failure}=\frac{1}{D^{N-1}}\sum_{i=1}^{D^{N-1}} \xi_{i}
\ee
This connects the failure to win directly with the order parameters which measure how close $\vert \phi \rangle$ is to having the necessary stabilizers.

%%%%%%%%%%%%%%%%%%%%%%%
\section{Discussion and related developments\label{sect:discussion}}
%%%%%%%%%%%%%%%%%%%%%%%
In this work we developed an understanding of entanglement in terms of local unitary groups of stabilizers.  This provides entanglement with
a clear physical meaning, as corresponding to unitary transformations that act on disjoint parts of a system but are nevertheless equivalent
to each other.  It also provides a clear mathematical setting for characterizing entanglement, in terms of quotients of the stabilizer group and
its subgroups.  We explored the general structure of entanglement groups and showed that monogamy of entanglement is a consequence.
We explored the extension to mixed states and showed that for a separable state, although the entanglement group need not be trivial, it can always be understood as arising
from entanglement with a purifying system.  Finally, we showed that this perspective on entanglement underlies several well-known quantum tasks.

As in section \ref{sect:classifying}, this understanding allows us to classify types of entanglement in terms of local unitary stabilizers, where two
local unitary stabilizer groups are identified if they are the same up to conjugation by a local unitary transformation.  The idea that entanglement is invariant under local
unitary transformations has been a theme in the literature and plays a role in other approaches to classifying types of entanglement.
In the following subsections we describe two approaches to classifying entanglement, that could be characterized as `algebraic' (section \ref{algebra}) and `geometric' (section
\ref{strata}), and we
make a connection to the classification in terms of stabilizer groups.  Building on this, we discuss methods for quantifying entanglement in section \ref{quantify}, and finally we
give some speculation about the general structure of entanglement groups in section \ref{structure}.

As a direction for future development, note that in this work we only considered two states to have identical entanglement if they are related by a local unitary transformation.
As an alternative, one could consider a coarser classification and consider two states to have identical entanglement if they are stochastically equivalent
under local operations and classical communication (SLOCC).  This amounts to equivalence under local $SL(d,{\mathbb C})$ transformations \cite{D_r_2000}.
It would be interesting to define and explore properties of entanglement groups that are invariant under the SLOCC classification.

%%%%%%%%%%%%%%%%%%%%%%
\subsection{Algebra of density matrices\label{algebra}}
%%%%%%%%%%%%%%%%%%%%%%
Consider a system of $N$ qubits.  We could group some of the qubits to form larger systems, but for the following discussion this is not
necessary.

We developed a description of entanglement in terms of the local unitary stabilizer group $S_{1 \cdots N}$.  It's perhaps more common to think
about entanglement in terms of the reduced density matrices associated with subsystems.  With any collection of spins $i_1 \cdots i_k$
we can associate a reduced density matrix $\rho_{i_1 \cdots i_k}$ obtained by tracing over the complementary spins.  How are these
descriptions related?

We can look at all of the reduced density matrices formed from a state $\ket{\psi}$, including the density matrix one gets just from the original
state without tracing.  To assemble these density matrices into a larger structure we first tensor each reduced density matrix with the
identity operator on the complementary qubits, to obtain a set of Hermitian matrices acting on the full Hilbert space. We can build
an even larger structure by taking commutators \cite{Kim:2021tse} and linear combinations of these operators until they close on a Lie algebra.  We call
the Lie algebra generated in this way the density matrix algebra.  It exponentiates to form a group which we call the density matrix
group $G_{DM}$.

What group does one get?  Suppose the state $\ket{\psi}$ has a local unitary stabilizer $u_1 \otimes \cdots \otimes u_N$.  The pure state density matrix
$\ket{\psi} \bra{\psi}$ is invariant under
conjugation by this stabilizer, and once we tensor with the identity on the complementary qubits, so are all of the reduced density matrices.\footnote{Having a
local stabilizer is crucial so the invariance survives partial tracing.}  Thus $G_{DM}$ will be invariant under conjugation by the stabilizer group $S_{1 \cdots N}$.
The stabilizer group has a centralizer $C(S_{1 \cdots N})$ inside the full unitary group $U(2^N)$ and we see that $G_{DM}$ is a subgroup of this centralizer.  For generic
states we don't expect $G_{DM}$ to be a proper subgroup, so most likely $G_{DM} = C(S_{1 \cdots N})$.  Conversely, given $G_{DM}$, we can compute its
centralizer.  The intersection of $C(G_{DM})$ with the local unitary group $U(2)^{N}$ gives the stabilizer group $S_{1 \cdots N}$.\footnote{A local unitary
stabilizer commutes with $G_{DM}$, so $S_{1 \cdots N} \subset C(G_{DM}) \cap U(2^N)$.  A local unitary $u \in C(G_{DM})$ must commute with $\ket{\psi}\bra{\psi}$,
but this is only possible if $u \ket{\psi} = \ket{\psi}$, so $S_{1 \cdots N} \supset C(G_{DM}) \cap U(2^N)$.}

In this sense characterizing entanglement using density matrices is equivalent to characterizing entanglement using local unitary stabilizers. Constructing
the density matrix algebra requires linear algebra, while finding the symmetries of a state is a highly non-linear problem, so the former may be simpler.
In the context of field theory, in place of the density matrix algebra, we can consider the algebra generated by the modular Hamiltonians of all subregions.
Indeed the algebra of CFT modular Hamiltonians is conjectured to be related to the structure of the bulk spacetime in holographic theories \cite{Kabat:2018smf},
and the bulk spacetime is believed to embody the entanglement structure of the CFT state \cite{VanRaamsdonk:2010pw}.

%%%%%%%%%%%%%%%%%%%%%%%%%%
\subsection{Strata and $LU$ equivalence classes\label{strata}}
%%%%%%%%%%%%%%%%%%%%%%%%%%
A seemingly different approach to classifying states according to entanglement was initiated in \cite{Linden_1998, Carteret_2000}.  The idea is that entanglement, whatever
it may be, does not change under local unitary transformations.
One can take the space of states of $N$ qubits ${\mathbb C} {\mathbb P}\big({2^N - 1}\big)$ and quotient
by the group of local unitary transformations $PSU(2)^{N}$.\footnote{We are following the alternate approach mentioned in footnote \ref{footnote:alternate}
where we quotient by a projective unitary group.}\footnote{Classification using the SLOCC group $SL(2,{\mathbb C})$ has also been discussed in the literature \cite{D_r_2000, Verstraete_2002}.}  This gives a space
\be
{\cal S}_N = {\mathbb C} {\mathbb P}\big({2^N - 1}\big) / PSU(2)^{N}
\ee
known as a stratified space \cite{Michel_2001}.  It has a complicated topological structure.

Each point in ${\cal S}_N$ represents an equivalence class of states that are related by local unitary transformations.  A state may have a local unitary stabilizer, which
means it is fixed under a subgroup of the local unitary group.  When a group action has a fixed point the quotient space is singular, so ${\cal S}_N$ has
singularities and breaks up into different orbit types or strata.  Each component of the space can be labeled by the stabilizer group that is preserved by those states \cite{Johansson_2014}.
In other words, strata can be labeled by entanglement groups.  Some simple examples are illustrated in Fig.\ \ref{fig:strata}.

\begin{figure}
\begin{center}
\hbox{\hspace{-11mm} \raisebox{1cm}{\includegraphics[width=8.2cm]{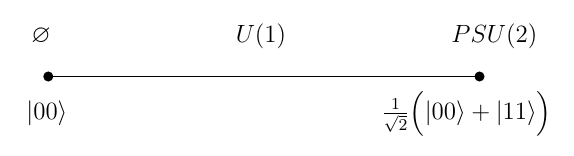}} \qquad \includegraphics[width=9.8cm]{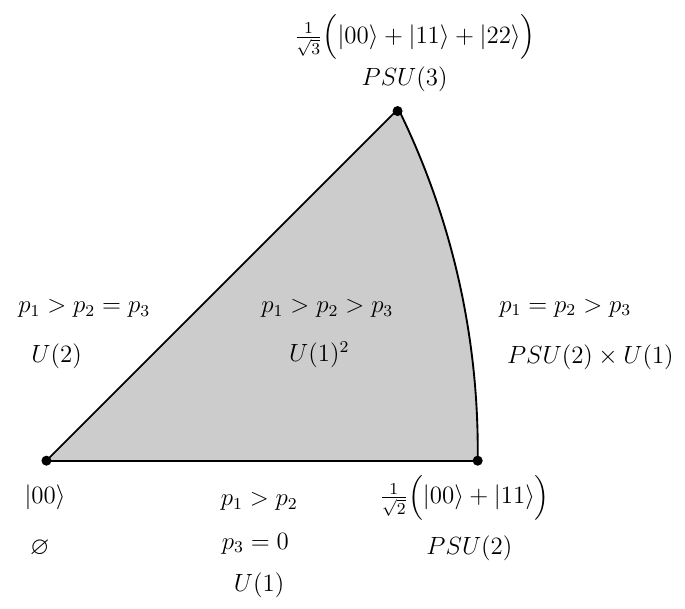}}
\end{center}
\caption{Strata for two qubits (left) and two qutrits (right).  The components are labeled by entanglement groups, with $\varnothing$ indicating the trivial group.
For bipartite systems the construction is easy thanks to the Schmidt decomposition (\ref{Schmidt}).  For qubits there is one parameter $0 \leq p_2 \leq 1/\sqrt{2}$.
For qutrits there are two parameters, $p_2$ (horizontal axis in figure) and $p_3$ (vertical axis in figure).  The parameters are constrained by $p_1 \geq p_2 \geq p_3 \geq 0$
where $p_1 = \sqrt{1 - p_2^2 - p_3^2}$.  We've written states at the extreme corners explicitly.
\label{fig:strata}}
\end{figure}

In this way we see that the state space of a quantum system can be decomposed into orbits of the group of local unitary transformations.  The entanglement groups control the way in which these orbits degenerate, and in this sense they have a geometric
representation in the state space.

Generically there is a principle stratum of states which have no local unitary stabilizer and there are lower-dimensional strata which have (in our language) some
entanglement relative to the partition into $N$ separate systems.\footnote{For bipartite systems even the principal stratum has a non-trivial entanglement group.}
We can classify entanglement according to which stratum a particular state is located on.  For the
action of a Lie group on a compact  manifold it is known that the number of different strata is finite \cite{Mostow_1957}, hence the classification is finite.  
This topological classification of states, according to which component ${\cal S}_N$ they belong to, is clearly closely connected to the idea of classifying entanglement via entanglement groups.

There is also a connection to the density matrix algebra discussed above.  Choose a point $s \in {\cal S}_N$ and imagine that one wants to flow to other points in the
same stratum.  To do this one can turn to the density matrix algebra which is a collection of Hermitian operators that commute with the stabilizer of $s$.  One can
choose one of these operators and think of it as a Hamiltonian.  It generates a flow along the stratum that preserves the
stabilizer group.  One may eventually reach a state with an enhanced stabilizer.  This happens when a stabilizer of the original state, which was an element of $U(2^N)$
but not in $U(2)^{\otimes N}$, becomes under time evolution along the flow, a local unitary stabilizer of the new state.

%%%%%%%%%%%%%%%%%%%%%%%%%%%
\subsection{Quantifying entanglement\label{quantify}}
%%%%%%%%%%%%%%%%%%%%%%%%%%%
So far we have given a discrete classification scheme for entanglement, in terms of the stratum that a state lies on which in turn determines its entanglement group.  One might ask for a more fine-grained method
for quantifying the entanglement of a state, in terms of one or more continuous parameters.  In the $g$-entanglement approach, these parameters should quantify how close or far a particular state is from nearby states
that have enhanced stabilizer groups.  These parameters can be physically meaningful, as we now outline.

The basic idea is that, although the entanglement group is the
same everywhere on the stratum, physical observables may be sensitive to other properties of the state that vary continuously as the state is varied.  In particular physical observables may be sensitive
to the entanglement of other nearby states.  To illustrate, consider an observable which is invariant under local unitary
transformations (so that it is well-defined on ${\cal S}_N$) and depends on the entanglement of a state.  As one moves along
a stratum the entanglement group does not change, but one may approach a different stratum that has a larger (say) entanglement group.
The observable could be sensitive to this, in which case a relevant parameter is how far the state is from the enhanced entanglement point.
For example, for two qubits, it is sensible to label points along the principle stratum in terms of distance from the Bell state.  This distance would control
the probability of success for tasks that depend on having enhanced $PSU(2)$ entanglement, in much the same way that the behavior of a many-body system
can be controlled by the existence of a nearby critical point.  In general the geodesic distances from strata with enhanced entanglement provide a possible
set of coordinates to use on ${\cal S}_N$.

One is free to adopt coordinates that have a direct physical meaning.  To illustrate this, consider the super-dense coding discussed in section \ref{sect:superdense}.  One could
attempt to do super-dense coding starting from a general state $\ket{\psi} = a \ket{00} + b \ket{11}$ by applying the operators (\ref{denseops}) to the first
qubit.  This produces four states proportional to the combinations $a \ket{00} \pm b \ket{11}$, $a \ket{10} \pm b \ket{01}$.  These states are generically linearly independent but not orthogonal.  One can
take $0 \leq b \leq 1/\sqrt{2}$ as a coordinate on ${\cal S}_2$.  The angle between the vectors increases monotonically with $b$, from $0^\circ$ at the tensor product point to $90^\circ$
at the Bell state, so the probability of successfully encoding a two-classical-bit message increases monotonically from $50\%$ at the tensor product point, where
there are pairs of messages that can't be distinguished, to $100\%$ at the Bell state, where there are four orthogonal states available to encode the message.  Thus the probability of
successful encoding provides a continuous parameter which can be used to characterize the entanglement of a two-qubit state.

In general, the appropriate notion of distance from an enhanced stabilizer may depend on the task at hand.  For example, for a system of $N$ qubits, the probability of failure in the GHZ game
can be used to provide a notion of distance from an enhanced stabilizer.  This example is discussed in sections \ref{subsub:GHZ} and \ref{subsub:generalGHZ}.

%%%%%%%%%%%%%%%%%%%%%%%%%%%
\subsection{Structure of entanglement groups\label{structure}}
%%%%%%%%%%%%%%%%%%%%%%%%%%%
It would be desirable to have a more complete understanding of the structure of the entanglement groups.  As a hint in this direction,
consider the GHZ state (\ref{GHZstate}).  The three-party entanglement group is $E_{ABC} = \lbrace \, 1,\,x \, \rbrace$
where $x$ is the additional stabilizer defined in (\ref{GHZuAuBuC}).  Note that $E_{ABC}$ is a subgroup of $S_{ABC}$ and that $x$ acts
on the two-party stabilizers as an automorphism.  That is, with $s$ denoting one of the matrices in (\ref{ghz2p1}), (\ref{ghz2p2}), (\ref{ghz2p3}),
the map $s \rightarrow x s x^{-1}$ takes $\phi \rightarrow - \phi$ and shifts the phases $\theta$.  This means the stabilizer group for GHZ is a
semi-direct product.
\be
S_{ABC} = (S_{AB} \cdot S_{AC} \cdot S_{BC}) \rtimes E_{ABC}
\ee
If we quotient both sides by $S_A \times S_B \times S_C$ this becomes a statement about the full entanglement group for GHZ, namely that it
is a semi-direct product of two-party and three-party entanglement.
\be
F_{ABC} = (E_{AB} \cdot E_{AC} \cdot E_{BC}) \rtimes E_{ABC}
\ee
It is tempting to conjecture that full entanglement groups always have a structure of this form.  For example it is tempting to speculate
that with four parties we would have
\be
F_{ABCD} = (F_{ABC} \cdot F_{ABD} \cdot F_{ACD} \cdot F_{BCD}) \rtimes E_{ABCD}
\ee
with each full three-party group having the decomposition given above.

%%%%%%%%%%%%%%%%%%%
\bigskip
\goodbreak
\centerline{\bf Acknowledgements}
\noindent
We are grateful to Dorit Aharonov for many valuable discussions and suggestions and to Phuc Nguyen for collaboration at an early stage.
The work of XJ, DK and AM was supported by U.S.\ National Science Foundation grant PHY-2112548.
GL was supported in part by the Israel Science Foundation under grant 447/17.

\appendix
%%%%%%%%%%%%%%%%%%%%%%%
\section{General bipartite systems\label{appendix:bipartite}}
%%%%%%%%%%%%%%%%%%%%%%%
Here we consider general bipartite systems.  We begin with the general formalism then determine the entanglement groups explicitly using a Schmidt decomposition
of the state.  There are no surprises here.  Instead our goal is to illustrate the formalism and make contact with results in the literature.

For a bipartite system the Hilbert space is a tensor product, ${\cal H}={\cal H}_{A} \otimes {\cal H}_{B}$.  We denote the dimensions $d_A$, $d_B$
respectively.  We fix a pure state $\ket{\psi} \in {\cal H}$ and ask for the stabilizer groups that leave the state invariant up to a phase.  The goal is to
construct the entanglement group
\be
E_{AB} = S_{AB} / \big( S_A \times S_B \big)
\ee

For bipartite systems we can explicitly determine the stabilizer groups by making a Schmidt decomposition of the state.  There are orthonormal bases for
${\cal H}_A$ and ${\cal H}_B$ such that
\bea
\label{Schmidt}
&& \ket{\psi} = \sum_{i = 1}^r \, p_i \,\ket{i}_A \otimes \ket{i}_B \\
\nonumber
&& p_1 \geq p_2 \geq \cdots \geq p_r > 0
\eea
Here $r$ is the Schmidt rank of the state.  The state corresponds to a $d_A \times d_B$ matrix
\be
\label{matrix}
\ket{\psi} \leftrightarrow \left(\begin{array}{ccccc}
p_1 & & & & \\
& \ddots & & & \\[-2pt]
& & p_r & & \\
& & & 0 & \\[-4pt]
& & & &\ddots
\end{array}\right)
\ee
with $u_A$ acting from the left and $u_B$ acting from the right in the complex conjugate representation.
The Schmidt coefficients may have degeneracies.
\be
\big(p_1,\ldots,p_r\big) = \big(\underbrace{\lambda_1,\ldots,\lambda_1}_{k_1},\underbrace{\lambda_2,\ldots,\lambda_2}_{k_2},\ldots,
\underbrace{\lambda_n,\ldots,\lambda_n}_{k_n}\big)
\ee
From this we identify the stabilizer groups
\bea
\nonumber
&& S_{AB} = U(1) \times U(k_1) \times \cdots \times U(k_n) \times U(d_A - r) \times U(d_B - r) \\
\label{bipartiteS}
&& S_A = U(1) \times U(d_A - r) \\
\nonumber
&& S_B = U(1) \times U(d_B - r)
\eea
The $U(1)$ factors in these groups correspond to the overall phases in (\ref{uAuB}), (\ref{uA}), (\ref{uB}).  The factors $U(d_A - r)$ and $U(d_B - r)$
act on the block of zeroes in (\ref{matrix}) from the left and right, respectively.  Taking the quotient, we identify the two-party entanglement group as
\be
\label{bipartiteEAB}
E_{AB} = PSU(k_1) \times U(k_2) \times \cdots U(k_n)
\ee
where in $E_{AB}$ we're taking the first factor to be an element of $PSU(k_1) = U(k_1) / U(1)$.\footnote{For the final quotient by an overall
$U(1)$ phase this is an arbitrary but convenient choice of a representative of each equivalence class.}  For bipartite systems the full stabilizer group is a direct product.
\be
S_{AB} = S_A \times S_B \times E_{AB}
\ee
There is a hierarchy to bipartite entanglement which is fully characterized by the Schmidt coefficients.  A tensor product state has ($p_1 = 1$, all others zero) and the entanglement
group is trivial. A maximally-entangled state has full Schmidt rank $r = {\rm min}(d_A,d_B)$ and all $p_i$'s equal; the entanglement group is $PSU(d_r)$.
Note that for bipartite systems there is always entanglement ($E_{AB}$ is non-trivial) unless the system is in a tensor product state.

For a simple explicit situation consider two qubits ($d_{A}=d_{B}=2$).  In a Schmidt basis the general state is
\begin{equation}
|\psi \rangle=a|00 \rangle+b|11 \rangle
\end{equation}
with $a \geq b \geq 0$.  There are three possibilities.
\begin{itemize}
\item  The tensor product state $\ket{00}$ with $a = 1$, $b=0$.  One-party stabilizers $S_A$, $S_B$ correspond to unitary matrices
\be
u_A = e^{i \theta_A} \left(\begin{array}{ll} 1 & 0 \\ 0 & e^{i\alpha} \end{array}\right)
\qquad
u_B = e^{i \theta_B} \left(\begin{array}{ll} 1 & 0 \\ 0 & e^{i\beta} \end{array}\right)
\ee
so that $S_A = U(1)^2$, $S_B = U(1)^2$.  There are no non-trivial two-party stabilizers, in other words $S_{AB} = S_A \times S_B = U(1)^4$, which means the entanglement group
is trivial, $E_{AB} = \lbrace 1 \rbrace$.

 \item The generic state ($a > b > 0$).  Two-party stabilizers $S_{AB}$ correspond to unitary matrices
\begin{equation}
\label{genericuAuB}
u_A \otimes u_B =
e^{i \theta} \begin{pmatrix}
e^{i \alpha_1} & 0 \\
0 &e^{i\alpha_2}
\end{pmatrix}
\otimes
\begin{pmatrix}
e^{-i \alpha_1} & 0 \\
0 & e^{-i\alpha_2}
\end{pmatrix}
\end{equation}
so that $S_{AB} = U(1)^3$.  One-party stabilizers are just phases times the identity operator, in other words $S_A = U(1)$, $S_B = U(1)$, which means
the entanglement group is $E_{AB} = U(1)$.

\item The Bell state with $a=b=1/\sqrt{2}$.  Two-party stabilizers $S_{AB}$ correspond to unitary matrices
\begin{equation}
\label{BelluAuB}
u_A \otimes u_B =
e^{i \theta} \begin{pmatrix}
\alpha & \beta \\
\gamma &\delta
\end{pmatrix}
\otimes
\begin{pmatrix}
\bar{\alpha} & \bar{\beta} \\
\bar{\gamma} & \bar{\delta}
\end{pmatrix} \qquad\quad \begin{pmatrix}
\alpha & \beta \\
\gamma &\delta
\end{pmatrix} \in U(2)
\end{equation}
(an overall phase times the ${\bf 2} \times \overline{\bf 2}$ of $U(2)$).
So the two-party stabilizer group is $S_{AB} = U(1) \times U(2)$.  One-party stabilizers are just phases times the identity operator, in other words $S_A = U(1)$, $S_B = U(1)$, which means the entanglement
group is $E_{AB} = PSU(2) = U(2) / U(1)$.
\end{itemize}

%%%%%%%%%%%%%%%%%%%%%%%%%%%%%%%
\section{Generalized Goursat lemma\label{appendix:Goursat}}
%%%%%%%%%%%%%%%%%%%%%%%%%%%%%%%
The Goursat lemma is standard result concerning subgroups of $A \times B$ \cite{Goursat,AndersonCamillo}.  The extension to subgroups of a product of three groups is
more involved and has been discussed in \cite{bauer2015generalized,neuen2016subgroups}.  Here we begin with a result for certain quotients of subgroups of
$A \times B \times C$ that will be useful for us, stated in Theorem \ref{asymmetric}.  A consequence relevant to two-party entanglement is given in Corollary \ref{two-party},
and the result directly relevant to three-party entanglement is given in Theorem \ref{symmetric}.

Consider a product group $A \times B \times C$.  We introduce projections, for example
\be
\begin{array}{ll}
\pi_A \, : \, A \times B \times C \rightarrow A & (a,b,c) \mapsto a \\[2pt]
\pi_{AB} \, : \, A \times B \times C \rightarrow A \times B & (a,b,c) \mapsto (a,b)
\end{array}
\ee
where $a \in A$, $b \in B$, $c \in C$.
Consider a subgroup $G \subset A \times B \times C$.  Define
\be
G_A = \pi_A(G) = \lbrace a \in A \, \vert \, \hbox{\rm $(a,b,c) \in G$ for some $b \in B$, $c \in C$}\rbrace
\ee
with analogous definitions for $G_B$, $G_C$.  Note that $G_A$ is a subgroup of $A$.  We also define
\bea
\nonumber
&& S_{AB} = \lbrace (a,b,1) \in G \, \rbrace \\
&& S_{A} = \lbrace (a,1,1) \in G \, \rbrace
\eea
with analogous definitions for $S_{AC}$, $S_{BC}$, $S_B$, $S_C$.  This gives us a collection of normal subgroups.  For example with the convenient
notation $G = S_{ABC}$ we have
\be
S_A \, \triangleleft \, S_{AB},\quad S_{AB} \, \triangleleft \, S_{ABC},\quad S_A \, \triangleleft \, S_{ABC}
\ee
and so on.  To see this one just has to check that for example conjugating an element of $S_A$ by an element of $S_{AB}$ gives another element of $S_A$.

The product of normal subgroups gives another normal subgroup, so we have a variety of building blocks at our disposal.  For the moment we'll focus on the
particular combination
\be
\label{N}
N = S_{AC} \cdot S_{BC} \, \triangleleft \, G
\ee
Projecting on the first or second factor we have
\bea
\nonumber
&& N_A = \pi_A(N) \, \triangleleft \, G_A \\
&& N_B = \pi_B(N) \, \triangleleft \, G_B
\eea
We are interested in comparing the quotients $G_A / N_A$ and $G_B / N_B$.  These are isomorphic, as we summarize with the following theorem.

\begin{theorem}
\label{asymmetric}
Consider a group $G \subset A \times B \times C$.  Let $N = S_{AC} \cdot S_{BC}$ where
\bea
\nonumber
&& S_{AC} = \lbrace (a,1,c) \in G \, \rbrace \\
&& S_{BC} = \lbrace (1,b,c) \in G \, \rbrace
\eea
$N$ is a normal subgroup, $N \, \triangleleft \, G$.  Denoting projections by
\be
G_A = \pi_A(G), \quad N_A = \pi_A(N), \quad G_B = \pi_B(G), \quad N_B = \pi_B(N)
\ee
the quotient groups $G_A / N_A$ and $G_B / N_B$ are isomorphic, $G_A / N_A \approx G_B / N_B$.
\end{theorem}

\begin{proof}
Define a map $\alpha \, : \, G_A \rightarrow G_B / N_B$ as follows.  Let $a \in G_A$.  Choose an element $b \in G_B$ such that $(a,b) \in \pi_{AB}(G)$.
Set $\alpha(a) = b N_B$.

We first have to check that $\alpha$ is a well-defined map from $A$ to $G_B / N_B$.  Suppose $(a,b,c) \in G$ and $(a,b',c') \in G$.  Taking the inverse and
multiplying we see that
\be
(1,b b'^{-1}, c c'^{-1}) \in G
\ee
But this is an element of $S_{BC}$ and therefore an element of $N$, which means $b b'{}^{-1} \in N_B$ or equivalently $b \in b' N_B$.
In other words $b$ and $b'$ are in the same coset so the map well-defined.

Note that $\alpha$ is a homomorphism since
\be
\hbox{\rm $(a_1,b_1,c_1) \in G$ and $(a_2,b_2,c_2) \in G$ $\Rightarrow$ $(a_1a_2,b_1b_2,c_1c_2) \in G$}
\ee
and therefore
\be
\alpha(a_1) \alpha(a_2) = [b_1] [b_2] = [b_1 b_2] = \alpha(a_1a_2)
\ee
Moreover $\alpha$ is surjective since every $b \in G_B$ is the image of some $a \in G_A$.  Finally we show that ${\rm ker} \, \alpha = N_A$.  First let's show ${\rm ker} \, \alpha \subset N_A$.  Suppose
$a \in {\rm ker} \, \alpha$.  Then there is an element $(a,b,c) \in G$ with $b \in N_B$.  Now $N_B = \pi_B(N) = \pi_B(S_{BC})$, which means $(1,b,c') \in S_{BC}$
for some $c'$.  Taking the inverse and multiplying we see that $(a,1,cc'^{-1}) \in G$.  This acts as the identity on system $B$, so it's an element of $S_{AC}$ and
therefore an element of $N$.  So $a \in N_A$ which means ${\rm ker} \, \alpha \subset N_A$.  Conversely if $a \in N_A$ then there is
some $b \in N_B$ for which $(a,b) \in \pi_{AB}(N)$; this means $\alpha(a) = N_B$ so $a \in {\rm ker} \, \alpha$ which shows that $N_A \subset {\rm ker} \, \alpha$.  Putting these statements together
we have ${\rm ker} \, \alpha = N_A$.

It then follows from the first isomorphism theorem that $G_A / {\rm ker} \, \alpha = G_A / N_A \approx G_B / N_B$.
\end{proof}

Theorem \ref{asymmetric} is an extension of the usual Goursat lemma in the following sense.
Any group $G \subset A \times B$ can be extended to a group $i(G) \subset A \times B \times C$
by appending the identity element on $C$.
\be
i(G) = \lbrace (a,b,1) \in A \times B \times C \, \vert \, (a,b) \in G \rbrace
\ee
For this extended group we have
\be
S_{AC} = S_A,\quad S_{BC} = S_B
\ee
which means $N = S_{AC} \cdot S_{BC}$ is a direct product, $N = S_A \times S_B$.  In this way we recover the usual Goursat isomorphisms
for $G \subset A \times B$ from theorem \ref{asymmetric}, namely that
\be
{\pi_A(G) \over \pi_A(S_A)} \approx {\pi_B(G) \over \pi_B(S_B)}
\ee

For our discussion of two-party entanglement, it's useful to note the following \cite{AndersonCamillo}.  Let $G \subset A \times B$ and let $N \, \triangleleft \, G$.  Denote the projections $G_A = \pi_A(G)$,
$N_A = \pi_A(N)$, etc.\ and define
\be
\label{H}
H = \lbrace \, \big( [a],\,[b] \big) \in {G_A \over N_A} \times {G_B \over N_B} \, \Big\vert \, (a,b) \in G \, \rbrace
\ee
Then $H \approx G/N$.  To see this note that the map $\phi \, : \, G \rightarrow H$, $\phi\big((a,b)\big) = \big([a],[b]\big)$ is a
surjective homomorphism with ${\rm ker} \, \phi = N$.  The first isomorphism theorem then establishes that $G / N \approx H$.
Moreover note that
the pairs of equivalence classes appearing in (\ref{H}) are related by the isomorphism established in theorem \ref{asymmetric},
so any $t \in G/N$ acts isomorphically on systems $A$ and $B$.

We summarize this discussion, relevant to two-party entanglement, as

\begin{corollary}
\label{two-party}
Consider a group $G \subset A \times B$.  Let $N = S_A \times S_B$ where
\bea
\nonumber
&& S_A = \lbrace \, (a,1) \in G \, \rbrace \\
&& S_B = \lbrace \, (1,b) \in G \, \rbrace
\eea
Then $N$ is a normal subgroup, $N \, \triangleleft \, G$, and we have the isomorphisms
\be
\label{isomorphisms2}
G / N \approx G_A / N_A \approx G_B / N_B
\ee
where $G_A = \pi_A(G)$, $N_A = \pi_A(N)$, etc.  As a group $G/N$ is diagonally embedded in $(G_A/N_A) \times (G_B/N_B)$,
that is,
\be
\label{diagonal2}
G/N \approx \lbrace \, \big(\theta_A(t), \theta_B(t) \big) \, \big\vert \, t \in G / N \, \rbrace
\ee
where $\theta_A$ and $\theta_B$ are isomorphisms.
\end{corollary}

In our discussion of three-party entanglement, a variant of theorem \ref{asymmetric} will be more useful for us.  It was essential for the
proof that we quotient by $S_{AC} \cdot S_{BC}$.  We can however take some additional quotients without spoiling
the result.  We will show that the same basic argument goes through if we replace $N$ in (\ref{N}) with
\be
N = S_{AB} \cdot S_{AC} \cdot S_{BC} \, \triangleleft \, G
\ee
The additional quotient by $S_{AB}$ doesn't spoil the isomorphism $G_A / N_A \approx G_B / N_B$, but we
gain a new isomorphism $G_A / N_A \approx G_C / N_C$.  The argument for the diagonal embedding in (\ref{diagonal2})
also goes through.  We summarize this as

\begin{theorem}
\label{symmetric}
Consider a group $G \subset A \times B \times C$.  Let $N = S_{AB} \cdot S_{AC} \cdot S_{BC}$ where
\bea
\nonumber
&& S_{AB} = \lbrace \, (a,b,1) \in G \, \rbrace \\
&& S_{AC} = \lbrace \, (a,1,c) \in G \, \rbrace \\
\nonumber
&& S_{BC} = \lbrace \, (1,b,c) \in G \, \rbrace
\eea
Then $N$ is a normal subgroup, $N \, \triangleleft \, G$.  Denoting $G_A = \pi_A(G)$, $N_A = \pi_A(N)$, etc., we have
the isomorphisms
\be
\label{isomorphisms3}
G / N \approx G_A / N_A \approx G_B / N_B \approx G_C / N_C
\ee
Moreover $G/N$ is diagonally embedded in $(G_A/N_A) \times (G_B/N_B) \times (G_C/N_C)$, meaning
\be
\label{diagonal3}
G/N \approx \lbrace \, \big(\theta_A(t), \theta_B(t), \theta_C(t) \big) \, \big\vert \, t \in G / N \, \rbrace
\ee
where $\theta_A$, $\theta_B$, $\theta_C$ isomorphisms.
\end{theorem}

\begin{proof}
We first adapt the proof of theorem \ref{asymmetric} to show that $G_A / N_A \approx G_B / N_B$.  Most of the steps go through but there is some extra work in showing that
${\rm ker} \, \alpha \subset N_A$.  To this end suppose $a \in {\rm ker} \, \alpha$.  Then there is an element $(a,b,c) \in G$ with $b \in N_B$.
In the new setting $N_B = \pi_B(N) = \pi_B(S_{AB} \cdot S_{BC})$, so there are elements $(a_1,b_1,1) \in S_{AB}$ and $(1,b_2,c_2) \in S_{BC}$
with $b_1 b_2 = b$.  Taking the product of these elements gives an element $x = (a_1,b,c_2) \in G$.  Multiplying $(a,b,c)$ by $x^{-1}$ gives
$y = (a a_1^{-1},1,c c_2^{-1}) \in G$.
This acts as the identity on system $B$, so in fact $y \in S_{AC}$.  Finally multiplying $y$ by $(a_1,b_1,1)$ we obtain
\be
z = (a,b_1,cc_2^{-1}) \in S_{AB} \cdot S_{AC}
\ee
This is an element of $N$, so $a \in N_A$ which establishes that ${\rm ker} \, \alpha \subset N_A$.

To complete the proof note that the new definition of $N$ is symmetric on $A$, $B$, $C$, so we have a three-way isomorphism
\be
G_A / N_A \approx G_B / N_B \approx G_C / N_C
\ee
Then to establish (\ref{diagonal3}) note that the argument around (\ref{H}) applies to the new setting provided we make the obvious definitions
\bea
\nonumber
&& H = \lbrace \, \big( [a],\,[b],\,[c] \big) \in {G_A \over N_A} \times {G_B \over N_B} \times {G_C \over N_C} \, \Big\vert \, (a,b,c) \in G \, \rbrace \\
&& \phi \, : \, G \rightarrow H, \quad \phi\big((a,b,c)\big) = \big([a],[b],[c]\big)
\eea
\end{proof}

Finally note that the usual Goursat lemma establishes a $1 \, : \, 1$ correspondence.  Given a group $G \subset A \times B$ one can construct $N_A \triangleleft G_A$
and $N_B \triangleleft G_B$ as in corollary \ref{two-party} and one has the isomorphisms (\ref{isomorphisms2}).  Conversely given a collection of groups satisfying
$N_A \triangleleft G_A$, $N_B \triangleleft G_B$ and an isomorphism $\theta \, : \, G_A / N_A \rightarrow G_B / N_B$ one can reconstruct $G$ using
\be
G = \lbrace \, (a,b) \in G_A \times G_B \, \big\vert \, \theta([a]) = [b] \, \rbrace
\ee
It would be interesting to understand if this converse can be extended to the entanglement groups of a tripartite system.  Is information about the two-party and three-party
entanglement groups enough to reconstruct the full entanglement group $F_{ABC}$?

%%%%%%%%%%%%%%%%%%%%%%%%%%%%%
\section{Two-party stabilizers for a separable state\label{appendix:twoparty}}
%%%%%%%%%%%%%%%%%%%%%%%%%%%%%
Here we consider the purification of a separable density matrix given in (\ref{pure}).
\be
\label{pure3}
\vert \psi \rangle = \sum_\ell \sqrt{p_\ell} \, \vert \ell \rangle_A \otimes \vert \ell \rangle_B \otimes \vert \ell \rangle_C
\ee
We want to understand what separability implies for a two-party stabilizer that acts on systems $A$ and $B$.

To this end, suppose there's a stabilizer $s_{AB} = u_A \otimes u_B \otimes \identity_C \in S_{AB}$.  This means
\be
\label{uAuB1}
(u_A \otimes u_B \otimes \identity_C) \vert \psi \rangle = e^{i \theta} \vert \psi \rangle \qquad \hbox{\rm for some phase $\theta$}
\ee
What can we say about the structure of the stabilizer, given the form of the state $\vert \psi \rangle$?  By acting on (\ref{uAuB1}) with ${}_C\langle \ell \vert$ we see that
\be
\label{uAuBonl}
u_A \vert \ell \rangle_A \otimes u_B \vert \ell \rangle_B = e^{i \theta} \vert \ell \rangle_A \otimes \vert \ell \rangle_B \qquad \hbox{\rm same phase $\theta$ for every $\ell$}
\ee
If a non-zero vector $x$ can be written as a tensor product in two different ways, say $x = y_1 \otimes y_2 = z_1 \otimes z_2$, it follows that the vectors are proportional,
$y_1 = c z_1$ and $y_2 = {1 \over c} z_2$.  Since $u_A$ and $u_B$ are unitary they can't change the magnitude of a vector, but they can introduce phases.  So from (\ref{uAuBonl}) we must have
\be
u_A \vert \ell \rangle_A = e^{i \alpha_\ell} \vert \ell \rangle_A \qquad u_B \vert \ell \rangle_B = e^{i \beta_\ell} \vert \ell \rangle_B \quad {\rm with} \quad \alpha_\ell + \beta_\ell = \theta
\ee
In other words $\vert \ell \rangle_A$, $\vert \ell \rangle_B$ are eigenvectors of $u_A$, $u_B$ with eigenvalues that are related by the condition $\alpha_\ell + \beta_\ell = \theta$.

Next we deal with the fact that the vectors $\vert \ell \rangle_A$, $\vert \ell \rangle_B$ may not span all of ${\cal H}_A$, ${\cal H}_B$.  To do this define
\bea
\nonumber && {\cal H}_{A_1} = {\rm span} \, \lbrace \vert \ell \rangle_A \rbrace \hspace{7mm} \hbox{\rm (the vector space spanned by $\vert \ell \rangle_A$)} \\
\nonumber && {\cal H}_{A_2} = \left({\cal H}_{A_1}\right)_\perp \hspace{1.3cm} \hbox{\rm (the orthogonal complement of ${\cal H}_{A_1}$ inside ${\cal H}_A$)} \\
&& {\cal H}_{B_1} = {\rm span} \, \lbrace \vert \ell \rangle_B \rbrace  \hspace{7mm} \hbox{\rm (the vector space spanned by $\vert \ell \rangle_B$)} \\
\nonumber && {\cal H}_{B_2} = \left({\cal H}_{B_1}\right)_\perp \hspace{1.3cm} \hbox{\rm (the orthogonal complement of ${\cal H}_{B_1}$ inside ${\cal H}_B$)}
\eea
The vectors $\vert \ell \rangle_A$ form an (overcomplete) basis for ${\cal H}_{A_1}$.  Since $u_A$ preserves this basis, it follows that $u_A$ is block-diagonal
when acting on ${\cal H}_{A_1} \oplus {\cal H}_{A_2}$.
\be
u_A = \left(\begin{array}{c|c}
u_{A_1} & 0 \\
\hline
0 & u_{A_2}
\end{array}\right)
\ee
Likewise $u_B$ preserves ${\cal H}_{B_1}$, so it is block-diagonal when acting on ${\cal H}_{B_1} \oplus {\cal H}_{B_2}$.
\be
u_B = \left(\begin{array}{c|c}
u_{B_1} & 0 \\
\hline
0 & u_{B_2}
\end{array}\right)
\ee
At this point it's convenient to diagonalize $u_{A_1}$.  Suppose $u_{A_1}$ has distinct eigenvalues
\bea
\nonumber
&& e^{i a_1} \quad \hbox{\rm with degeneracy $m_1$} \\
&& e^{i a_2} \quad \hbox{\rm with degeneracy $m_2$} \\
\nonumber
&& \qquad\qquad \vdots \\
\nonumber
&& e^{i a_s} \quad \hbox{\rm with degeneracy $m_s$}
\eea
We likewise diagonalize $u_{B_1}$ with distinct eigenvalues $e^{i b_i}$ and corresponding degeneracies $n_i$, $i = 1,\ldots,s$.
The eigenvalues are related by the condition $a_i + b_i = \theta$.  (This condition implies that $u_{A_1}$ and $u_{B_1}$ have the same
number of distinct eigenvalues.)  In this diagonal basis $u_A$ and $u_B$ have the form
\bea
\label{uAuBform}
&& u_A = \left(\begin{array}{ccc|c}
e^{i a_1} \identity_{m_1 \times m_1} &&& \\
& \ddots && \\
&& e^{i a_s} \identity_{m_s \times m_s} & \\
\hline
&&& u_{A_2}
\end{array}\right) \\[8pt]
\nonumber
&& u_B = \left(\begin{array}{ccc|c}
e^{i b_1} \identity_{n_1 \times n_1} &&& \\
& \ddots && \\
&& e^{i b_s} \identity_{n_s \times n_s} & \\
\hline
&&& u_{B_2}
\end{array}\right)
\eea
Now we turn to system $C$.  Let's order the labels $\ell$ so that
\bea
\nonumber
&& \hbox{\rm the first $o_1$ have $u_A \vert \ell \rangle_A = e^{i a_1} \vert \ell \rangle_A$ and $u_B \vert \ell \rangle_B = e^{i b_1} \vert \ell \rangle_B$} \\
&& \hbox{\rm the next $o_2$ have $u_A \vert \ell \rangle_A = e^{i a_2} \vert \ell \rangle_A$ and $u_B \vert \ell \rangle_B = e^{i b_2} \vert \ell \rangle_B$} \\
\nonumber
&& \qquad\quad \vdots \\
\nonumber
&& \hbox{\rm the last $o_s$ have $u_A \vert \ell \rangle_A = e^{i a_s} \vert \ell \rangle_A$ and $u_B \vert \ell \rangle_B = e^{i b_s} \vert \ell \rangle_B$}
\eea
Then we can introduce stabilizers
\bea
&& s_{AC} = e^{i \theta'} \, u_A \otimes \identity_B \otimes \left(\begin{array}{ccc}
e^{-i a_1} \identity_{o_1 \times o_1} && \\
& \ddots & \\
&& e^{-i a_s} \identity_{o_s \times o_s}
\end{array}\right) \\[8pt]
\nonumber
&& s_{BC} = e^{i \theta''} \, \identity_A \otimes u_B \otimes \left(\begin{array}{ccc}
e^{-i b_1} \identity_{o_1 \times o_1} && \\
& \ddots & \\
&& e^{-i b_s} \identity_{o_s \times o_s}
\end{array}\right)
\eea
that by construction satisfy
\be
s_{AC} \vert \psi \rangle = e^{i \theta'} \vert \psi \rangle \qquad\quad s_{BC} \vert \psi \rangle = e^{i \theta''} \vert \psi \rangle
\ee
If we choose the phases to satisfy $\theta' + \theta'' = \theta$, then we also have
\be
\label{CommonCenter}
s_{AC} s_{BC} = s_{AB}
\ee
In other words, for a state obtained by purifying a separable density matrix, an $AB$ stabilizer can always be thought of as a combination of an $AC$ stabilizer with a $BC$ stabilizer.

We can turn this statement about stabilizers into a statement about the two-party pure-state entanglement that is present in the purification.
Thinking of $s_{AC}$, $s_{BC}$, $s_{AB}$ as representatives of equivalence classes in the pure-state entanglement groups $E_{AC}$, $E_{BC}$, $E_{AB}$, namely
\be
e_{AC} = [s_{AC}] \in E_{AC}, \quad e_{BC} = [s_{BC}] \in E_{BC}, \quad e_{AB} = [s_{AB}] \in E_{AB}
\ee
we have
\be
\label{CommonCenter2}
e_{AC} e_{BC} = e_{AB}
\ee
That is, for a state obtained by purifying a separable density matrix, $AB$ entanglement can always be thought of as a combination of $AC$ entanglement and $BC$ entanglement.

Entanglement satisfying (\ref{CommonCenter2}) is allowed, but as discussed in section \ref{sect:noshare}, it takes a very special and restricted form.  In brief, by projecting (\ref{CommonCenter2}) onto system $A$,
we see that $g_A = \pi_A(e_{AC})$ is an element of both $\pi_A(E_{AC})$ and $\pi_A(E_{AB})$.  In section \ref{sect:noshare} we showed that distinct two-party entanglement groups
must commute element-by-element, which means that $g_A$ must be in the center of both groups.\footnote{Similar reasoning leads to (\ref{Z}).}
\be
g_A \in Z\big(\pi_A(E_{AB})\big) \cap Z\big(\pi_A(E_{AC})\big)
\ee
Moreover, the isomorphism theorem (\ref{IsomorphismReference})
implies that elements of the common center $Z\big(\pi_A(E_{AB})\big) \cap Z\big(\pi_A(E_{AC})\big)$ act isomorphically
on systems $A$, $B$, $C$.

%%%%%%%%%%%%%%%%%%%%%%%%%%%%%
\section{Werner state example\label{appendix:Werner}}
%%%%%%%%%%%%%%%%%%%%%%%%%%%%%
To illustrate the controlled unitary discussed in section (\ref{sect:A(BC)}), we use the example of the Werner state \cite{Werner:1989zz}, defined for two qubits as
\begin{equation}
W = p \, \dyad{\text{Bell state}} + \frac{1-p}{4} \, \identity_{4\times 4}
\end{equation}
where $p \in [0,1]$. But first, let's look at the entanglement group.  The two-party stabilizer group $\widetilde{S}_{AB}$ consists of matrices
\begin{equation}
\widetilde{s}_{AB} = e^{i \theta} \begin{pmatrix}
\alpha & \beta \\
\gamma &\delta
\end{pmatrix}
\otimes
\begin{pmatrix}
\bar{\alpha} & \bar{\beta} \\
\bar{\gamma} & \bar{\delta}
\end{pmatrix} \qquad\quad \begin{pmatrix}
\alpha & \beta \\
\gamma &\delta
\end{pmatrix} \in U(2)
\end{equation}
(an overall phase, times the tensor product of a unitary matrix with its complex conjugate).
To see that this is a stabilizer, note that $\widetilde{s}_{AB} \vert \, \hbox{\rm Bell state} \rangle = e^{i \theta} \vert \, \hbox{\rm Bell state} \rangle$, so the two terms in $W$ are
separately invariant under conjugation by $\widetilde{s}_{AB}$.  This means the two-party stabilizer group is $\widetilde{S}_{AB} = U(1) \times U(2)$.  One-party stabilizers
$\widetilde{s}_A$, $\widetilde{s}_B$ are just phases times the identity operator, so $\widetilde{S}_A = \widetilde{S}_B = U(1)$.  The entanglement group is therefore
\be
\EtildeAB = \big(U(1) \times U(2)\big) / \big(U(1) \times U(1)\big) = U(2) / U(1) = PSU(2)
\ee

The Werner state is separable when $0 \leq p \leq 1/3$ and non-separable for $1/3 < p \leq 1$.
Our goal is to display a block-diagonal unitary matrix (\ref{block}) which completely disentangles system $B$, but
only in the range $0 \leq p \leq 1/3$ where the Werner state is separable.

We begin by presenting the Werner state in the form
\bea
\nonumber
W & = & \frac{p}{2} \, \Big(\dyad{00}{00} + \dyad{11}{11} + \dyad{--} + \dyad{++} + \dyad{\circ\times} + \dyad{\times\circ} \Big) \\
\label{Werner}
& & + (1 - 3p) \, \frac{1}{4} \, \identity_{4\times 4}
\eea
where $+-$ is the $x$-basis and $\times\circ$ is the $y$-basis: $\ket{+}=\big(\ket{0}+\ket{1}\big)/\sqrt{2}$ and $\ket{\times} = \big(\ket{0}+i\ket{1})\big/\sqrt{2}$, etc.  This presents the state as a sum of
tensor products with positive coefficients, except for the the last term where it turns negative precisely when $p > 1/3$.  Assuming the state is separable ($p \leq 1/3$), it can be purified by introducing
an auxiliary 10 dimensional Hilbert space, to
\bea
\nonumber
\ket{\psi} & = & \sqrt{\frac{p}{2}}\Big(\ket{000}+\ket{111}+\ket{--2}+\ket{++3}+\ket{\circ\times4}+\ket{\times\circ5}\Big) \\
& & + \sqrt{\frac{1-3p}{4}}\Big(\ket{006}+\ket{017}+\ket{108}+\ket{119}\Big), \label{werner_sep_pure}
\eea
Note that this purification is not minimal.  It is only a purification of the Werner state if $p \leq 1/3$.  (If $p > 1/3$,
there are imaginary coefficients, and tracing out the auxiliary system will produce a coefficient $\left \vert {1 - 3p \over 4}
\right \vert$ that differs by a sign from (\ref{Werner}).) On the enlarged Hilbert space, we can write down a block-diagonal unitary by inspection:
\bea
\nonumber
&& {\cal U}_{BC} = \identity_A\otimes\Big[\Big(\ketbra{0}{1}+\ketbra{1}{0}\Big)\otimes \dyad{1} + \Big(\ketbra{0}{-}+\ketbra{1}{+}\Big)\otimes\dyad{2} \\
\nonumber
&&  + \Big(\ketbra{0}{+}+\ketbra{1}{-}\Big)\otimes\dyad{3} + \Big(\ketbra{0}{\times}+\ketbra{1}{\circ}\Big)\otimes\dyad{4} + \Big(\ketbra{0}{\circ}+\ketbra{1}{\times}\Big)\otimes\dyad{5} \\
&& +\Big(\ketbra{0}{1}+\ketbra{1}{0}\Big)\otimes\Big(\dyad{7}+\dyad{9}\Big)\Big] + \identity_A \otimes \identity_B \otimes \Big(\dyad{0} + \dyad{6} + \dyad{8}\Big)
\eea
Applying this unitary takes (\ref{werner_sep_pure}) to the claimed tensor product form.
\begin{equation}
\label{werner_factor}
{\cal U}_{BC}\ket{\psi} = \ket{0}_B\otimes\Big[ \sqrt{\frac{p}{2}}\Big(\ket{00}+\ket{11}+\ket{-2}+\ket{+3}+\ket{\circ4}+\ket{\times5}\Big) + \sqrt{\frac{1-3p}{4}}\Big(\ket{06}+\ket{07}+\ket{18}+\ket{19}\Big) \Big]_{AC}
\end{equation}
Of course the Werner state can be purified for any value of $p$.  Indeed a minimal purification, valid for any $0 \leq p \leq 1$, is
\be
\label{w_min_pure}
\ket{\psi} = \sqrt{\frac{1-p}{4}} \, \left[ {1 \over \sqrt{2}} \big( \ket{00} - \ket{11} \big) \otimes \ket{0}
+ \ket{01} \otimes \ket{1} + \ket{10} \otimes \ket{2} \right]
+ \sqrt{\frac{1+3p}{8}} \, \big( \ket{00} + \ket{11} \big) \otimes \ket{3}
\ee
If $p \leq 1/3$ this can be related to (\ref{werner_sep_pure}) by a unitary on system $C$ and can be mapped to
(\ref{werner_factor}) by ${\cal U}_{BC}$.  However if $p > 1/3$ the Werner state is not separable and no such unitary
on $C$ exists.

%%%%%%%%%%%%%%%%%%%%%%%%%%%%%
\section{Entanglement groups for ${\cal U}_{BC} \vert \psi \rangle$ \label{appendix:UBC}}
%%%%%%%%%%%%%%%%%%%%%%%%%%%%%
As shown in section \ref{sect:A(BC)}, for a separable state, there is a controlled unitary ${\cal U}_{BC}$ that completely disentangles system $B$.  This means we can define a new state
\be
\vert \psi' \rangle = {\cal U}_{BC} \vert \psi \rangle
\ee
that has the form given in (\ref{factorized}).
\be
\vert \psi' \rangle = \vert \chi \rangle_B \otimes \Big(\sum_{\ell = 1}^L \sqrt{p_\ell} \, \vert \ell \rangle_A \otimes \vert \ell \rangle_C\Big)
\ee
The entanglement groups for $\vert \psi' \rangle$ have the following properties.
\begin{enumerate}
\item
Since $\vert \psi \rangle$ and $\vert \psi' \rangle$ are related by a (block diagonal) unitary transformation on the combined $BC$ system, the $A$ -- $BC$ entanglement groups are unchanged.
\be
E_{A(BC)}^{\vert \psi \rangle} = E_{A(BC)}^{\vert \psi' \rangle}
\ee
\item
The $A$ -- $BC$ entanglement group for $\vert \psi' \rangle$ acts trivially on system $B$.
\be
\label{prop2}
E_{A(BC)}^{\vert \psi' \rangle} = E_{AC}^{\vert \psi' \rangle}
\ee
To see this, note that the Schmidt decomposition of the state $\vert \psi' \rangle$ has the form\footnote{This is the Schmidt decomposition for a division into systems $A$
and $BC$.  It's the same as the Schmidt decomposition one would obtain by projecting the state $\vert \psi' \rangle$ into ${\cal H}_A \otimes {\cal H}_C$.}
\be
\vert \psi' \rangle = \sum_{i = 1}^r \sqrt{p_r} \, \vert i \rangle_A \otimes \vert \chi \rangle_B \otimes \vert i \rangle_C
\ee
where $r$ is the rank of $\rho_A = {\rm Tr}_B \rho_{AB}$.  For the purposes of $A$ -- $BC$ entanglement, this means we can restrict our attention to stabilizers of the form
\be
s_{AC}^{\vert \psi' \rangle} = u_A \otimes \identity_B \otimes u_C
\ee
It follows that $E_{A(BC)}^{\vert \psi' \rangle} = E_{AC}^{\vert \psi' \rangle}$.  For the original state $\vert \psi \rangle$ note that the relevant stabilizers
$s_{A(BC)}^{\vert \psi \rangle}$ do act on system $B$, but in a way that is completely fixed by
\be
s_{A(BC)}^{\vert \psi \rangle} = {\cal U}_{BC}^\dagger s_{AC}^{\vert \psi' \rangle} {\cal U}_{BC}
\ee
\item
Since $\vert \psi' \rangle$ is a tensor product, there is no entanglement between $B$ and the combined $AC$ system.
\be
\label{prop3}
E_{B(AC)}^{\vert \psi' \rangle} = \lbrace 1 \rbrace
\ee
\end{enumerate}
These features of the entanglement groups for $\vert \psi' \rangle$ are consequences of having a separable density matrix.

Given (\ref{prop3}), the argument is reversible.  Consider a state $\vert \psi' \rangle$, which without loss of generality we write as\footnote{Any state can be expanded in
a basis of tensor product states, $\vert \psi' \rangle = \sum c_{i\ell} \vert i \rangle_{AB} \otimes \vert \ell \rangle_C$.  Collect the coefficient of $\vert \ell \rangle_C$ and write
it as a unit vector $\vert \ell \rangle_{AB}$ times a magnitude $\sqrt{p_\ell}$.}
\be
\vert \psi' \rangle = \sum_\ell \sqrt{p_\ell} \, \vert \ell \rangle_{AB} \otimes \ell_C
\ee
with ${}_C \langle \ell \vert \ell' \rangle_C = \delta_{\ell \ell'}$.
If $E_{B(AC)}^{\vert \psi' \rangle} = \lbrace 1 \rbrace$, then $\vert \psi' \rangle$ must be a tensor product of the form (\ref{factorized}).  Then, by applying
the inverse of the block-diagonal unitary transformation (\ref{block}) we obtain (\ref{pure2}), which is the general purification of a separable density matrix.

%%%%%%%%%%%%%%%%%%%%%%%%%%%%%%%
%\bibliographystyle{utphys}
%\bibliography{entanglement}
\providecommand{\href}[2]{#2}\begingroup\raggedright\endgroup

\end{document}